\documentclass[10pt,reqno]{amsart}
\usepackage{amssymb,amscd,graphicx,stmaryrd,mathrsfs,bbm,accents,color,comment,mathtools}
\usepackage
{hyperref}
\newtheorem{theorem}{Theorem}[section] 
\newtheorem{proposition}[theorem]{Proposition} 
\newtheorem{corollary}[theorem]{Corollary} 
\newtheorem{lemma}[theorem]{Lemma}

\newtheorem{example}[theorem]{Example}
 
\begin{document}
\title[Hastings factorization]{On Hastings factorization for quantum many-body systems in the infinite volume setting}
\author[A.\ Ukai]{Ayumi Ukai}
\address{
Graduate School of Mathematics, 
Nagoya University, 
Nagoya, 464-8602, Japan
}
\email{ukai.ayumi.t8@s.mail.nagoya-u.ac.jp}
\date{\today}

\begin{abstract}
We will give self-contained and detailed proofs of the (multidimensional) Hastings factorization as well as a ($1$-dimensional) area law in a wider setup than previous works. Especially, they are applicable to both quantum spin and fermion chains of infinite volume as special cases. 
\end{abstract}

\maketitle

\allowdisplaybreaks

\tableofcontents

\section{Introduction and Ganeral setting}\label{S1}
Taku Matsui \cite{Matsui13} pointed out that Hastings's proof \cite{Hastings07} of an area law in the setting of $1$-dimensional spin chains of finite volume essentially works even in the infinite volume setup. Here, an \emph{area law} means that the entanglement entropy of a given state with respect to a finite region, say $X$, grows at most propotionally with the size of boundary of $X$. Actually, he gave just a very brief outline to prove the so-called Hastings factorization (for the density operator of a gapped ground state in its GNS representation) and established the following: 
\begin{center}
Hastings factorization $\Rightarrow$ Area law ($\Rightarrow$ Split property)
\end{center} 
for $1$-dimensional quantum spin chains of infinite volume. After almost a decade, Matsui \cite{Matsui20} also gave a way to derive the split property for fermion chains of infinite volume from the corresponding result for quantum spin chains by utilizing the Jordan--Wigner transformation. 

Several years after Matsui's work \cite{Matsui13}, Yoshiko Ogata \cite{Ogata20} found a remarkable application of Matsui's split property. Actually, she used it to give a mathematically rigorous definition to the SPT (symmetry protected topological phase) index in the setting of $1$-dimensional quantum spin chains of infinite volume. See her recent exposition \cite{Ogata22} too. This suggests that any facts related to an area law have the potential to play an important r\^{o}le in the study of quantum many-body systems of infinite volume. 

Although the notion of projected entangled pair states (PEPS) is believed to be a multidimensional counterpart of matrix product states (MPS), we want to reconsider the question of what the right multidimensional analog of MPS is from the viewpoint of Hastings's proof of a ($1$-dimensional) area law. Since the Hastings factorization can be regarded as an MPS-like approximation of a gapped ground state, we think that a multidimensional generalization of the Hastings factorization may play an important r\^{o}le in investigating the question. The same motivation was also mentioned by Hamza, Michalakis, Nachtergaele and Sims \cite{Hamza-Michalakis-Nachtergaele-Sims}, but their result given there requires an assumption on local Hamiltonians rather than a ground state.  

The purpose of this paper is two-fold. The first one is to give a detailed proof of the Hastings factorization in the multidimensional and infinite volume setting with special emphasis on the part of treating unbounded operators arising from unbounded regions that never emerge in the finite volume setting. Here again, we would like to emphasize that Matsui \cite{Matsui13} gave just a very brief sketch and treated only the $1$-dimensional setting. The second one is to formulate a class of quantum many-body systems, for which Hastings's and Matsui's arguments work well. The class naturally includes both quantum spin and fermion systems so that this work provides a unified approach to a $1$-dimensional area law in the setting of infinite volume. Here we remind the reader that the possibility to prove a $1$-dimensional area law for fermion chains of infinite volume was already suggested in \cite{Matsui20} without any details.

The proofs concering ``Hastings factorization $\Rightarrow$ ($1$-dimensional) Area law" that we will give in this paper can be regarded as a reconstruction of Matsui's work \cite{Matsui13} \emph{in our wider setup}, though we need some care on a difficult point (see subsection \ref{S4.1}) in Matsui's discussion there. Thus, we do give a detailed proof for the sake of completeness. We also extend some part of the proof to the multidimensional setting (see subsection \ref{S4.1}). Unfortunately, we do not know at the present moment whether or not such a multidimensional generalization is actually useful, but we hope that our generalization clarifies which part of Hastings's and Matsui's proofs depend on the $1$-dimensional setting crucially. 
 
\medskip
We will explain the general framework of quantum many-body systems based on $C^*$-algebras.

Let $(\Gamma,d)$ be a (possibly infinite) countable discrete metric space, of which we think as the underlying space of a quantum many-body system in question. For two subsets $\Lambda_1, \Lambda_2 \subseteq \Gamma$ we define the distance $d(\Lambda_1,\Lambda_2)$ between them as usual, that is, $d(\Lambda_1,\Lambda_2) := \inf\{ d(x_1,x_2);\ x_i \in \Lambda_i, i=1,2\}$. We denote by $\mathcal{P}_0(\Gamma)$ all the finite subsets of $\Gamma$. 

Let $\mathcal{A}$ be a unital $C^*$-algebra, whose (self-adjoint) elements are regarded as \emph{observables}. Assume that $\mathcal{A}$ has an increasing net $\{\mathcal{A}_X\}_{X \in \mathcal{P}_0(\Gamma)}$ of simple $C^*$-subalgebras. We write $\mathcal{A}_\mathrm{loc} := \bigcup_{X\in\mathcal{P}_0(\Gamma)}\mathcal{A}_X$, whose (self-adjoint) elements are regarded as \emph{local observables}. Assume also that $\mathcal{A}_\mathrm{loc}$ is a dense subset in $\mathcal{A}$. Note that $\mathcal{A}$ becomes simple and hence any representation must be faithful. To each subset $\Lambda \subseteq \Gamma$, we assign the $C^*$-subalgebra $\mathcal{A}_\Lambda$ of $\mathcal{A}$ generated by all the $\mathcal{A}_X$ with $\Lambda \supseteq X \in \mathcal{P}_0(X)$. 

The interaction $\Phi$ of the quantum many-body system in question is given as a function $\Phi : \mathcal{P}_0(\Gamma) \to \mathcal{A}_\mathrm{loc}$ such that $\Phi(X) = \Phi(X)^* \in \mathcal{A}_X$ for each $X \in \mathcal{P}_0(\Gamma)$. 

The triple $(\mathcal{A},\{\mathcal{A}_X\}_{X \in \mathcal{P}_0(\Gamma)}, \Phi)$ we have introduced so far is indeed a general mathematical (or $C^*$-algebraic) framework of quantum many-body system in question. 
 
In this paper, we assume that all quantum many-body systems we will consider satisfy the condition:
\begin{itemize}
\item[(A1)] The interaction $\Phi$ satisfies $\Phi(X)\in\mathcal{A}_X\cap\mathcal{A}_{X^c}'$ for every $X\in\mathcal{P}_0(\Gamma)$.
\end{itemize}
This is fulfilled when the triple $(\mathcal{A},\{\mathcal{A}_X\}_{X \in \mathcal{P}_0(\Gamma)}, \Phi)$ is associated with a quantum spin or a fermion system (i.e., the interaction is an even one).

Under a suitable set of assumptions, the time evolution of the quantum many-body system in question is well defined as a strongly continuous $1$-parameter automorphism group $\alpha^t=\alpha^t_\Phi$ on $\mathcal{A}$. Then, the notion of ground states can be established as certain states on the $C^*$-algebra $\mathcal{A}$ with the help of the derivation associated with $\alpha^t$. 

For each ground state $\omega$, the Hamiltonian $H_\omega$ naturally emerges as a positive self-adjoint (possibly unbounded) operator on $\mathcal{H}_\omega$, where $\mathcal{A}\curvearrowright\mathcal{H}_\omega$ is the GNS representation associated with $\omega$. We consider $\mathcal{A}$ to be embedded in $\mathcal{B}(\mathcal{H}_\omega)$ and omit the symbols in the GNS representation since the representation is faithful. A ground state $\omega$ is said to be gapped if $H_\omega$ has spectral gap, that is, $0$ is a non-degenerate eigenvalue of $H_\omega$ and $\inf(\sigma(H_\omega)\setminus\{0\})> 0$ with the spectrum $\sigma(H_\omega)$. Remark that any gapped ground state must be pure (see e.g., \cite[Definition 4.2.6]{Sakai:book} and its explanation there).

\medskip\noindent
{\bf Acknowledgement.} We acknowledge our supervisor Professor Yoshimichi Ueda for his encouragements and editorial supports to this paper. We thank Professor Yoshiko Ogata for her intensive course at Nagoya in Dec., 2021, which initially motivated us to study this topic.  

\section{Preliminaries}\label{S2}
Throughout these notes, we keep the notation introduced in section \ref{S1}. In this section, we will provide the further notation and facts that will be used later. 

\medskip
For each $\Lambda\in\mathcal{P}_0(\Gamma)$, we define a local Hamiltonian 
\[
H_\Lambda=H_{\Phi,\Lambda}:=\sum_{X\subseteq\Lambda}\Phi(X)\in\mathcal{A}_\Lambda
\]
and the local time evolution to be a strongly continuous $1$-parameter automorphism group 
\[
\alpha_\Lambda^t=\alpha_{\Phi,\Lambda}^t:=\mathrm{Ad}\,e^{itH_\Lambda}
\] 
of $\mathcal{A}$, where $\mathrm{Ad}\,e^{itH_\Lambda}(A) =e^{itH_\Lambda}Ae^{-itH_\Lambda}$. Then, for every $A\in\mathcal{A}$ we define
\[
\delta_\Lambda(A)=\delta_{\Phi,\Lambda}(A):=i[H_\Lambda,A], 
\]
which coincides with 
\[
\frac{d}{dt}\Big|_{t=0}\alpha^t_\Lambda(A)
:=\lim_{t\to0}\frac{1}{t}(\alpha_\Lambda^t(A)-A)\quad\text{in norm}.
\] 
Moreover, 
\[
\alpha_\Lambda^t(A)=\exp(t\delta_\Lambda)(A)
:=\sum_{n=0}^\infty\frac{t^n}{n!}\delta_\Lambda^n(A)
\]
holds for every $A\in\mathcal{A}$.

\subsection{Lieb-Robinson bound}\label{S2.1}
We begin by recalling the notion of $F$-functions. 

An $F$-function $F: [0,\infty)\to(0,\infty)$ is a non-increasing function satisfying the following conditions: 
\begin{gather*}
\Vert F\Vert :=\sup_{x\in\Gamma}\sum_{y\in\Gamma} F(d(x,y))<\infty,\\
c_F :=\sup_{x,y\in\Gamma}\sum_{z\in\Gamma}\frac{F(d(x,z))F(d(z,y))}{F(d(x,y))} <\infty.
\end{gather*}
We set 
\[
\Vert\Phi\Vert_F :=\sup_{x,y\in\Gamma}\frac{1}{F(d(x,y))} \sum_{\substack{X\subseteq\Gamma\\ x,y\in X}}\Vert\Phi(X)\Vert.
\]

Recall that an interaction $\Phi$ is bounded, if
\[
\mathfrak{j}= \mathfrak{j}_\Phi
:= \sup_{x\in\Gamma}\sum_{x\in X\in\mathcal{P}_0(\Gamma)}\frac{\Vert\Phi(X)\Vert}{|X|}
<\infty,
\]
and that $\Phi$ has finite range $\mathfrak{r}=\mathfrak{r}_\Phi > 0$, if 
\[
\Phi(X)=0\quad\text{as long as}\quad\mathrm{diam}(X):=\sup\{ d(x_1,x_2)\,;\,x_1,x_2\in X\} >\mathfrak{r}.
\] 

Write $\mathfrak{B}_r(x):=\{ y\in\Gamma\,;\, d(x,y)\leq r\}$ for each pair $x\in\Gamma$ and $r>0$. We observe that
\[
\sup_{x\in\Gamma}|\mathfrak{B}_r(x)|\leq \sup_{x\in\Gamma}\sum_{y\in\Gamma} \frac{F(d(x,y))}{F(r)} =\frac{\Vert F\Vert}{F(r)},
\]
since we have $1\leq F(d(x,y))/F(r)$ for any $y\in\mathfrak{B}_r(x)$ and $0<F(d(x,y))/F(r)$ for any $y\in\Gamma$. Thus if $\Phi$ is bounded and has finite range $\mathfrak{r}$, we have 
\begin{align*}
\mathfrak{j}'= \mathfrak{j}'_\Phi 
&:= \sup_{x\in\Gamma}\sum_{x\in X\in\mathcal{P}_0(\Gamma)}\Vert\Phi(X)\Vert
\leq \frac{\mathfrak{j}\Vert F\Vert}{F(\mathfrak{r})}
<\infty, \\
\mathfrak{j}''= \mathfrak{j}''_\Phi 
&:= \sup_{x\in\Gamma}\sum_{x\in X\in\mathcal{P}_0(\Gamma)}|X|\Vert\Phi(X)\Vert
\leq \frac{\mathfrak{j}'\Vert F\Vert}{F(\mathfrak{r})}<\infty. 
\end{align*} 

The $\Phi$-boundary of an $X \in \mathcal{P}_0(\Gamma)$ is defined to be
\[
\partial_\Phi X :=\{x\in X\,;\,\text{$x\in Y$, $Y \cap X^c\neq\emptyset$ and $\Phi(Y)\neq0$ for some $Y\in\mathcal{P}_0(\Gamma)$}\}.
\] 

\medskip
Here is a version of the so-called Lieb-Robinson bound that we will need later. The idea of proof of \cite[Theorem 3.1]{Nachtergaele-Sims-Young18} or \cite[Theorem 4.3]{Bru-Siqueira Pedra:book} works for proving it. 

\begin{theorem}\label{thm_Lieb-Robinson}
Assume that $(\Gamma,d)$ admits an $F$-function and that $(\mathcal{A},\{\mathcal{A}_X\}_{X\in\mathcal{P}_0(\Gamma)},\Phi)$ is a quantum many-body system such that $\Vert\Phi\Vert_F<\infty$. Let $X,Y,\Lambda \in \mathcal{P}_0(\Gamma)$ with $X \subseteq Y \cap \Lambda$ be arbitrarily given. If $A\in\mathcal{A}_X$ and $B\in\mathcal{A}_Y' \cap \mathcal{A}$,
then we have
\[
\Vert[\alpha_\Lambda^t(A),B]\Vert
\leq \frac{2\Vert A\Vert\Vert B\Vert}{c_F} \left(e^{2\Vert\Phi\Vert_F c_F|t|} -1\right) \sum_{x\in\partial_\Phi X} \sum_{y\in Y^c} F(d(x,y)).
\]
\end{theorem}

\begin{corollary}[Lieb--Robinson bound for quantum many-body systems]\label{cor_Lieb-Robinson}
Assume that $(\Gamma,d)$ admits an $F$-function and that $(\mathcal{A},\{\mathcal{A}_X\}_{X\in\mathcal{P}_0(\Gamma)},\Phi)$ is a quantum many-body system such that $\Phi$ is bounded and has finite range. For any $\mu > 0$ there are two constants $c_\mu, v_\mu > 0$ so that the following holds true: Let $X,Y,\Lambda \in \mathcal{P}_0(\Gamma)$ with $X \subseteq Y \cap \Lambda$ be arbitrarily given. If $A\in\mathcal{A}_X$ and $B\in\mathcal{A}_Y' \cap \mathcal{A}$,
then we have
\[
\Vert[\alpha_\Lambda^t(A),B]\Vert \leq c_\mu\Vert A\Vert \Vert B\Vert\, |\partial_\Phi X|\, e^{-\mu(d(X,Y^c)-v_\mu|t|)}.
\]
\end{corollary}

A proof of these theorem and corollary will be given in Appendix \ref{Appendix_Lieb-Robinson} for the reader's convenience, because our setting is wider than the known results. 

\subsection{Time evolution $\alpha^t=\alpha^t_\Phi$}\label{S2.2}

Throughout this subsection, let $(\Gamma,d)$ and $(\mathcal{A},\{\mathcal{A}_X\}_{X\in\mathcal{P}_0(\Gamma)},\Phi)$ be as in Corollary \ref{cor_Lieb-Robinson}, and we will also use the notations in subsection \ref{S2.1} freely. 

The next proposition implies that a time evolution $\alpha^t$ is constructed out of $\Phi$ as mentioned in section \ref{S1}. Thanks to assumption (A1), the idea of proof of \cite[Lemma 4.4]{Bru-Siqueira Pedra:book} or \cite[Theorem 2.2]{Nachtergaele-Ogata-Sims} works for proving it.

\begin{proposition}\label{prop_time}
We obtain a strongly continuous $1$-parameter automorphism group $\alpha^t$ defined by 
\[
\lim_{\Lambda\nearrow\Gamma}\alpha_\Lambda^t(A)=\alpha^t(A)
\]
for each $A\in\mathcal{A}$ and each $t\in\mathbb{R}$, where $\lim_{\Lambda\nearrow\Gamma}$ denotes the thermodynamic limit (see Appendix \ref{Appendix_derivation}). Furthermore, the convergence of $\alpha^t_\Lambda(A)$ is uniform in $t$ on any compact subsets.
\end{proposition}

A detailed proof of this proposition will be given in Appendix \ref{Appendix_derivation} for the reader's convenience.

\medskip
The next lemma will be used later. 
\begin{lemma}\label{lem_time}
For any $A\in\mathcal{A}_X$ with $X\in\mathcal{P}_0(\Gamma)$ and for every $\Lambda\in\mathcal{P}_0(\Gamma)$ with $d(\Lambda^c,X)>\mathfrak{r}:=\mathfrak{r}_\Phi$, we have
\[
\Vert\alpha^t(A)-\alpha^t_\Lambda(A)\Vert \leq \frac{c_{\mu}\mathfrak{j}'}{\mu v_\mu} |\partial_\Phi X||\partial_\Phi\Lambda|\, \Vert A\Vert\, e^{-\mu(d(\Lambda^c,X)-\mathfrak{r})}(e^{\mu v_\mu|t|}-1).
\]
Here, $\mathfrak{j}'=\mathfrak{j}'_\Phi$, and moreover, $\mu, c_\mu, v_\mu$ are the constants that appear in Corollary \ref{cor_Lieb-Robinson}. 
\end{lemma}

\begin{proof}
Thanks to assumption (A1), in the same way as in \cite[Eq.(2.28)]{Nachtergaele-Ogata-Sims} (see Appendix \ref{Appendix_derivation}), for any $F$-function and any $\Lambda,\Lambda' \in\mathcal{P}_0(\Gamma)$ with $X\subseteq\Lambda\subseteq\Lambda'$ and $d(\Lambda^c,X)>\mathfrak{r}$, we have
\begin{align*}
&\left\Vert\left(\alpha_{\Lambda'}^t-\alpha_\Lambda^t\right)(A)\right\Vert\\
&\leq\frac{2\Vert A\Vert}{c_F} \int_{\min\{0,t\}}^{\max\{0,t\}} \big(e^{2\Vert\Phi\Vert_F c_F|s|}-1\big) \,ds \sum_{\substack{Z\subseteq\Lambda'\\ Z\cap\Lambda\neq\emptyset\\ Z\cap\Lambda^{c}\neq\emptyset}} \Vert \Phi(Z)\Vert \sum_{x\in\partial_\Phi X} \sum_{y\in \Lambda_\mathrm{int}(\mathfrak{r})^c} F(d(x,y)),
\end{align*}
where we write $\Lambda_\mathrm{int}(\mathfrak{r}):=\{x\in\Lambda\,;\,d(x,\Lambda^c) >\mathfrak{r}\}$.

We have a new $F$-function $F_\mu(r):=e^{-\mu r}F(r)$ (see e.g., \cite[Appendix 2]{Nachtergaele-Sims-Young19}). Since $\Phi$ is bounded and has finite range, we have $\Vert\Phi\Vert_{F_\mu}\leq e^{\mu\mathfrak{r}}\Vert \Phi\Vert_F<\infty$. Applying the above equation with $F=F_\mu$, we have 
\begin{align*}
\left\Vert\left(\alpha_{\Lambda'}^t-\alpha_\Lambda^t\right)(A)\right\Vert
&\leq\frac{2\Vert A\Vert}{c_{F_\mu}}  
\int_{\min\{0,t\}}^{\max\{0,t\}} \big(e^{2\Vert\Phi\Vert_{F_\mu}c_{F_\mu}|s|}-1\big) \,ds \\
&\quad\sum_{z\in\partial_\Phi\Lambda} \sum_{\substack{Z\in\mathcal{P}_0(\Gamma)\\ z\in Z}} \Vert \Phi(Z)\Vert
\sum_{x\in\partial_\Phi X} \sum_{y\in \Lambda_\mathrm{int}(\mathfrak{r})^c} e^{-\mu d(x,y)}F(d(x,y))\\
&\leq \frac{2\Vert A\Vert}{c_{F_\mu}} 
\frac{e^{2\Vert\Phi\Vert_{F_\mu} c_{F_\mu}|t|}-1}{2\Vert\Phi\Vert_{F_\mu} c_{F_\mu}}
\big(|\partial_\Phi \Lambda| \mathfrak{j}'\big)\,
\big(e^{-\mu d(\partial_\Phi X,\Lambda_\mathrm{int}(\mathfrak{r})^c)} |\partial_\Phi X| \Vert F\Vert\big).
\end{align*}
Taking $\Lambda'\nearrow\Gamma$, we have
\[
\left\Vert\left(\alpha^t-\alpha_\Lambda^t\right)(A)\right\Vert
\leq \frac{\mathfrak{j}'\Vert F\Vert}{\Vert\Phi\Vert_{F_\mu}c_{F_\mu}^2} |\partial_\Phi X| |\partial_\Phi \Lambda| \, \Vert A\Vert\, e^{-\mu(d(\Lambda^c,X)-\mathfrak{r})} (e^{2\Vert\Phi\Vert_{F_\mu} c_{F_\mu}|t|} -1),
\]
where we note that 
$d(\partial_\Phi X,\Lambda_\mathrm{int}(\mathfrak{r})^c)
\geq d(X,\Lambda_\mathrm{int}(\mathfrak{r})^c)
\geq d(X,\Lambda^c)-\mathfrak{r}$. Hence we have obtained the desired inequality with letting $v_\mu:=2\Vert\Phi\Vert_{F_\mu} c_{F_\mu}/\mu$ and $c_\mu:=2\Vert F\Vert/c_{F_\mu}$.
\end{proof}

The next proposition consists of several well-known facts and will play an important r\^{o}le later. The proofs in \cite[subsection 3.3.2]{Naaijkens13:book} or \cite[theorem 6.2.4]{Bratteli-Robinson:book} work for proving this proposition too thanks to Proposition \ref{prop_time} and assumption (A1). 

\begin{proposition}\label{prop_derivation}
In this context, the following statements hold:
\begin{itemize}
\item[(i)] For every $A\in\mathcal{A}_\mathrm{loc}$ 
\[
\delta(A):=\lim_{\Lambda\nearrow\Gamma}\delta_\Lambda(A)
\] 
coverges in norm and defines a (symmetric) derivation $\delta=\delta_\Phi :\mathcal{A}_\mathrm{loc}\to\mathcal{A}$.

\item[(ii)] For every $A\in\mathcal{A}_\mathrm{loc}$, there exists a $T>0$ such that
\[
\exp(t\delta)(A):=\sum_{n=0}^\infty\frac{t^n}{n!}\delta^n(A)
\]
converges uniformly in norm for any $t\in(-T,T)$ on any compact subsets. 

\item[(iii)] For every $A\in\mathcal{A}_\mathrm{loc}$, 
\[
\alpha^t(A)=\lim_{\Lambda\nearrow\Gamma}\alpha^t_\Lambda(A)=\lim_{\Lambda\nearrow\Gamma}\exp(t\delta_{\Lambda})(A)=\exp(t\delta)(A),
\]
uniformly in $t$ on any compact subset of $(-T,T)$. In particular, 
\[
\frac{d}{dt}\Big|_{t=0}\alpha^t(A)=\lim_{t\to0}\frac{1}{t}(\alpha^t(A)-A)=\delta(A)=\lim_{\Lambda\nearrow\Gamma}\delta_\Lambda(A)\quad\text{in norm}
\]
holds for every $A\in\mathcal{A}_\mathrm{loc}$. 
\end{itemize}
\end{proposition}

A detailed proof of this proposition will be given in Appendix \ref{Appendix_derivation} for the reader's convenience.

\medskip
Let $\omega$ be a ground state, and $\mathcal{A}\curvearrowright\mathcal{H}_\omega\ni\Omega$ be its GNS representation with GNS unit vector. It is known, see \cite[section 3.4]{Naaijkens13:book}, that we have a positive self-adjoint operator $H_\omega$ on $\mathcal{H}_\omega$ in such a way that $H_\omega\Omega=0$ and $\alpha^t=\mathrm{Ad}e^{itH_\omega}$ for any $t\in\mathbb{R}$. Here is a corollary, which is probably known among specialists. We give its proof for the sake of completeness.

\begin{corollary}\label{cor_dom}
With the above notation, $\mathcal{A}_\mathrm{loc}\mathrm{Dom}(H_\omega)\subseteq\mathrm{Dom}(H_\omega)$ and 
\[
[H_\omega,A]\phi=-i\delta(A)\phi=\lim_{\Lambda\nearrow\Gamma}[H_\Lambda,A]\phi
\]
holds for every pair $(A,\phi)\in\mathcal{A}_\mathrm{loc}\times\mathrm{Dom}(H_\omega)$. In particular, any $[H_\omega,A]$ with $A\in\mathcal{A}_\mathrm{loc}$ must be bounded and canonically extends to the whole $\mathcal{H}_\omega$. 
\end{corollary}

\begin{proof}
For every pair $(A,\phi)\in\mathcal{A}_\mathrm{loc}\times\mathrm{Dom}(H_\omega)$,
\[
\frac{1}{it}(e^{itH_\omega}-I)A\phi=\alpha^t(A)\frac{1}{it}(e^{itH_\omega}-I)\phi +\frac{1}{it}(\alpha^t(A)-A)\phi
\]
converges to $AH_\omega\phi-i\delta(A)\phi$ as $t\to0$ by Proposition \ref{prop_derivation}(iii). The desired assertion immediately follows by the Stone theorem and the construction of $\delta(A)$.
\end{proof}

\subsection{Approximation of observables by local observables}\label{S2.3}
Let $(\mathcal{A},\{\mathcal{A}_X\}_{X\in\mathcal{P}_0(\Gamma)},\Phi)$ be a quantum many-body system. Throughout this subsection, assume that a state $\omega$ is pure and that assumption (A2) holds below:
\begin{itemize}
\item[(A2)] For each $X\in\mathcal{P}_0(\Gamma)$, $\mathcal{A}_X$ is $*$-isomorphic to a full matrix algebra $M_{d_X}(\mathbb{C})$ with some $d_X\in\mathbb{N}$. 
\end{itemize}
Let $\mathcal{A}\curvearrowright\mathcal{H}_\omega$ be the GNS representation associated with $\omega$. Note that this representation must be irreducible. 

The next lemma is a folklore among operator algebraists. The keys are assumption (A2) and that $\mathcal{A}\curvearrowright\mathcal{H}_\omega$ is irreducible. 

\begin{lemma}\label{lem_commutant}
For any $X\in\mathcal{P}_0(\Gamma)$, $\mathcal{A}_X'= (\mathcal{A}_X'\cap\mathcal
{A})''$ holds on $\mathcal{H}_\omega$.
\end{lemma}

\begin{proof}
Choose a matrix unit system $\{E_{ij}\}_{1\leq i,j\leq d_X}$ of $\mathcal{A}_X \cong M_{d_X}(\mathbb{C})$ and let $\{e_i\}_{i=1}^d$ be a standard orthonormal basis of $\mathbb{C}^{d_X}$. Then we construct a unitary transform $W:\phi\in\mathcal{H}_\omega\mapsto\sum_{k=1}^{d_X} e_k\otimes E_{1k}\phi\in\mathbb{C}^{d_X}\otimes\mathcal{K}$ with $\mathcal{K}:=E_{11}\mathcal{H}_\omega$. We have 
\begin{align*} 
WA\phi 
&= 
\sum_{i=1}^{d_X} e_i \otimes E_{1i}A\phi 
= 
\Big(\sum_{i,j=1}^{d_X} |e_i\rangle\langle e_j| \otimes E_{1i}AE_{j1}\Big)W\phi 
\end{align*}
for any $A \in \mathcal{A}$ and $\phi\in\mathcal{H}_\omega$. Thus $W(\mathcal{A}_X \subseteq \mathcal{A})W^* = \big(M_{d_X}(\mathbb{C})\otimes\mathbb{C}I \subseteq M_{d_X}(\mathbb{C})\otimes E_{11}\mathcal{A}E_{11}\big)$. Hence we may identify $\mathcal{H}_\omega=\mathbb{C}^{d_X}\otimes\mathcal{K}$ and accordingly $(\mathcal{A}_X \subseteq \mathcal{A}) = (M_{d_X}(\mathbb{C})\otimes\mathbb{C}I \subseteq M_{d_X}(\mathbb{C})\otimes E_{11}\mathcal{A}E_{11})$. 

By irreduciblity, $M_{d_X}(\mathbb{C})\otimes(E_{11}\mathcal{A}E_{11})'' = \mathcal{A}'' = B(\mathbb{C}^{d_X}\otimes\mathcal{K}) = M_{d_X}(\mathbb{C})\otimes B(\mathcal{K})$ and hence $(\mathcal{A}_X'\cap\mathcal{A})'' = (\mathbb{C}I\otimes(E_{11}\mathcal{A}E_{11}))'' = \mathbb{C}I\otimes B(\mathcal{K}) = (M_{d_X}(\mathbb{C})\otimes\mathbb{C}I)' = \mathcal{A}_X'$.
\end{proof}

The next proposition provides a method of approximating observables by local observables, which will crucially be used later. It is probably a forklore among operator algebraists. 

\begin{proposition}\label{prop_partial_trace}
For any $X\in\mathcal{P}_0(\Gamma)$ there exists a (non-normal) conditional expectation (in particular, a positive contractive $\mathcal{A}_X$-bimodule map) $E :\mathcal{B}(\mathcal{H}_\omega)\to \mathcal{A}_X$ such that 
\[
\Vert A-E(A)\Vert\leq\sup_{U\in U(\mathcal{A}'_X)}\Vert[A,U]\Vert \leq \sup_{\substack{B\in\mathcal{A}'_X\cap\mathcal{A}\\ \Vert B\Vert\leq1}} \Vert[A,B]\Vert
\]
for any $A\in\mathcal{B}(\mathcal{H}_\omega)$. Here, $\mathcal{B}(\mathcal{H})$ denotes all the bounded linear operators on a Hilbert space $\mathcal{H}$ and $U(\mathcal{C})$ denotes the unitary group of a $C^*$-algebra $\mathcal{C}.$
\end{proposition}

\begin{proof}
Since $\mathcal{A}'_X$ has Schwartz's property P (see e.g., \cite[Theorem 2.3]{Ranard-Walter-Witteveen}), there is a conditional expectation $E :\mathcal{B}(\mathcal{H}_\omega)\to\mathcal{A}_X$ such that $E(A)$ falls into the intersection of $\mathcal{A}_X$ and the weak-operator closure of the convex hull of all $UAU^*$ with $U\in U(\mathcal{A}'_X)$ for every $A\in\mathcal{B}(\mathcal{H}_\omega)$. See e.g., \cite[Proposition 2]{Kadison04} for this folklore. Then, one can find a proof of the first inequality in \cite[page 3917]{Ranard-Walter-Witteveen}. 

By Lemma \ref{lem_commutant}, the second inequality follows thanks to the lower-semicontinuity of the norm in the weak operator topology and the Kaplansky density theorem.
\end{proof}

\section{Hastings factorization}\label{S3}
\subsection{Main statements}\label{S3.1}
\begin{theorem}[Hastings factorization]\label{thm_main}
Assume that $(\Gamma,d)$ admits an $F$-function and that $(\mathcal{A},\{\mathcal{A}_X\}_{X\in\mathcal{P}_0(\Gamma)},\Phi)$ is a quantum many-body system such that $\Phi$ is bounded and has finite range $\mathfrak{r} = \mathfrak{r}_\Phi$. For each $r\geq0$, we define the ``$r$-width boundary" of $X\in\mathcal{P}(\Gamma)$ 
\begin{equation}\label{eq_def-set}
\partial X(r) :=\{x\in X\,;\, d(x,X^c)\leq r\}\sqcup\{y\in X^c\,;\, d(y,X)\leq r\}, 
\end{equation}
where $d(x,X^c):=d(\{x\},X^c)$ and $d(y,X):=d(\{y\},X)$. 

Then, for any gapped ground state $\omega$, one can find two positive constants $C_1,C_2 >0$ that make the following hold: For any $X\in\mathcal{P}_0(\Gamma)$ and any $\ell>\mathfrak{r}$ there exist two projections $O_R(X,\ell)\in\mathcal{A}_X$, $O_L(X,\ell)\in\mathcal{A}_X'=(\mathcal{A}_X'\cap\mathcal{A})''\subseteq\mathcal{B}(\mathcal{H}_\omega)$ and a contraction $O_B(X,\ell)\in\mathcal{A}_{\partial X(3\ell+\mathfrak{r})}$ such that 
\[
\Vert O_B(X,\ell)O_L(X,\ell)O_R(X,\ell)-|\Omega\rangle\langle\Omega|\Vert
\leq C_1|\partial X(2\ell+\mathfrak{r})|^4 \exp(-C_2\ell),
\]
where $\Omega\in\mathcal{H}_\omega$ is the GNS unit vector of $\omega$ and $\mathcal{B}(\mathcal{H}_\omega)$ denotes all the bounded linear operators on $\mathcal{H}_\omega$.
\end{theorem}

\begin{corollary}\label{cor_main} 
Under the assumption of Theorem \ref{thm_main}, we assume further the following condition: 
\begin{itemize}
\item[(A3)] There exist $\nu\geq1$ and $\kappa>0$ such that $|\partial X(n+\mathfrak{r})|\leq \kappa|\partial X(\mathfrak{r})|n^\nu$ holds for any $X\in\mathcal{P}_0(\Gamma)$.
\end{itemize}

Then, for any gapped ground state $\omega$, one can find two positive constants $C_1,C_2 >0$ (independent of the choice of $(X,\ell)$) that make the following hold: For any $X\in\mathcal{P}_0(\Gamma)$ and any sufficiently large $\ell>0$ (n.b., how sufficiently large $\ell$ is depends on $|\partial X(\mathfrak{r})|$), there exist two projections $O_R(X,\ell)\in\mathcal{A}_X$, $O_L(X,\ell)\in\mathcal{A}_X'=(\mathcal{A}_X'\cap\mathcal{A})''\subseteq\mathcal{B}(\mathcal{H}_\omega)$ and a positive contraction $O_B(X,\ell)\in\mathcal{A}_{\partial X(3\ell+\mathfrak{r})}$ such that 
\[
\Vert O_B(X,\ell)O_L(X,\ell)O_R(X,\ell)-|\Omega\rangle\langle\Omega|\Vert
\leq C_1|\partial X(\mathfrak{r})| \exp(-C_2\ell),
\]
where $\Omega\in\mathcal{H}_\omega$ is the GNS unit vector of $\omega$.
\end{corollary}

We will call the consequence of the above corollary \emph{the Hastings factorization} in the rest of this paper. 

Corollary \ref{cor_main} is derived from Theorem \ref{thm_main} as explained below. Under the assumption of this corollary, we have
\[
C_1|\partial X(2\ell+\mathfrak{r})|^4 \exp(-C_2\ell)
\leq C_1(2^{\nu}\kappa)^4|\partial X(\mathfrak{r})|^4 \ell^{4\nu} \exp(-C_2\ell).
\]
Consider a sufficiently large $\ell>0$ such that
\[
\epsilon:=C_1(2^{\nu}\kappa)^4|\partial X(\mathfrak{r})|^4 \ell^{4\nu} \exp(-C_2\ell)<1
\]
holds. Then, as mentioned in \cite[Lemma 4]{Hastings18}, we have
\[
\Vert O_B(X,\ell)^*O_B(X,\ell)O_L(X,\ell)O_R(X,\ell)-|\Omega\rangle\langle\Omega|\Vert 
\leq \sqrt{2\epsilon}+3\epsilon+\epsilon^2
\leq 6\sqrt[4]{\epsilon}. 
\]
Hence, we have
\[
\Vert O_B(X,\ell)^*O_B(X,\ell)O_L(X,\ell)O_R(X,\ell)-|\Omega\rangle\langle\Omega|\Vert \leq C_1'|\partial X(\mathfrak{r})|\exp(-C'_2\ell)
\]
for some constants $C_1',C_2'$ independent of the choice of $(X,\ell)$. Therefore, $O_B(X,\ell)$ in Corollary \ref{cor_main} can be taken as a positive element by considering $O_B(X,\ell)$ instead of $O_B(X,\ell)^*O_B(X,\ell)$. 

\medskip
The assumptions for $(\Gamma,d)$ in Corollary \ref{cor_main} should be regarded as a generalization of the $\nu$-dimensional lattice $\mathbb{Z}^\nu$. Indeed, recall that $(\Gamma,d)$ is said to be \emph{$\nu$-regular} with $\nu\geq1$, if there exists a $\kappa>0$ such that 
\[
\sup_{x\in\Gamma}|\mathfrak{B}_r(x)|\leq\kappa r^\nu
\]
for any $r\geq1$. It is known, see \cite[Proposition A.1]{Nachtergaele-Sims-Young19}, that the $\nu$-regularity implies the existence of $F$-function and that the $\nu$-dimensional lattice $\mathbb{Z}^\nu$ with $\ell^1$-distance is $\nu$-regular. Moreover, we have
\[
\partial X(n+\mathfrak{r})\subseteq\bigcup_{x\in\partial X(\mathfrak{r})} \mathfrak{B}_n(x). 
\]
Hence the $\nu$-dimensional lattice $\mathbb{Z}^\nu$ with $\ell^1$-distance admits an $F$-function and satisfies assumption (A3). 

\subsection{Proof of Theorem \ref{thm_main}}
The strategy is here based on \cite{Hamza-Michalakis-Nachtergaele-Sims}. First, the Hamiltonian is approximated by three local observables $(M_{\chi},M_{\chi_{bd}},M_{\chi^c})$ as in \cite[Lemma 3.1]{Hamza-Michalakis-Nachtergaele-Sims}. Then the Hastings factorization is derived in the same way as in \cite[Chapter 4]{Hamza-Michalakis-Nachtergaele-Sims}, except for some technical issues based on the infinite volume setting. However, the construction of $M_{\chi^c}$ in \cite[Proposition 5.1]{Hamza-Michalakis-Nachtergaele-Sims} does not work, since the (bulk) Hamiltonian is not necessarily bounded. Instead, we refer to \cite{Matsui13}, though it gives only a rough sketch. The constructions of $M_{\chi}$ and $M_{\chi_{bd}}$ in \cite[Proposition 5.1]{Hamza-Michalakis-Nachtergaele-Sims} do work, except for some technical issues due to the infinite volume setting.

\medskip
We use the constants: $\mathfrak{r}=\mathfrak{r}_\Phi$, $\mathfrak{j}'=\mathfrak{j}'_\Phi$, $\mathfrak{j}''=\mathfrak{j}''_\Phi$ as well as $\mu, c_\mu, v_\mu$ in Corollary \ref{cor_Lieb-Robinson} (see subsection \ref{S2.1}). 

Let $\omega$ be a gapped ground state and $\mathcal{A}\curvearrowright\mathcal{H}_\omega\ni\Omega$ be its GNS representation. Let $H_\omega$ be the Hamiltonian, i.e., a positive self-adjoint operator on $\mathcal{H}_\omega$ with $H_\omega\Omega=0$. Write $\gamma :=\inf(\sigma(H_\omega)\setminus\{0\})>0$, the specural gap, and let 
\[
H_\omega=\int_0^\infty\lambda\,E_\omega(d\lambda)=\int_\gamma^\infty\lambda\,E_\omega(d\lambda)
\]
be the spectral decomposition of $H_\omega$. Since $H_\omega$ has spectral gap (and hence $0$ is a non-degenerate eigenvalue of $H_\omega$), we have $E_\omega(\{0\})=|\Omega\rangle\langle\Omega|$, which we will denote by $P_0$ for simplicity. 
The rest of this section is devoted to proving Theorem \ref{thm_main}. 

\medskip
In what follows, we choose and fix an arbitrary $X\in\mathcal{P}_0(\Gamma)$ and an arbitrary $n>0$, and define 
\begin{align*}
X_\mathrm{int}(n)
&:=\{x\in \Gamma\,;\, d(x,X^c)>n\},\\
X(n)
&:=X_\mathrm{int}(n)\sqcup \partial X(n)
=\{x\in \Gamma\,;\, d(x,X)\leq n\}.
\end{align*}
Here is a geometric lemma.
\begin{lemma}\label{lem_set}
For any $m,n\in\mathbb{N}$ with $m\geq n$, we have 
\[
|\partial_\Phi X|\leq|\partial X(\mathfrak{r})|,\quad
|\partial X(n)|\leq |\partial X(m)|. 
\]
Moreover, for any $n\in\mathbb{N}$, we have
\[
|\partial_\Phi(X_\mathrm{int}(n))|\leq |\partial X(n+\mathfrak{r})|,\quad
|\partial_\Phi(X(n+\mathfrak{r}))|\leq |\partial X(n+\mathfrak{r})|.
\]
\end{lemma}
\begin{proof}
By their definition, we have $\partial_\Phi\Lambda\subseteq \partial\Lambda(\mathfrak{r})$ for any $\Lambda\subseteq\Gamma$ and $\partial X(n)\subseteq\partial X(m)$ if $n\leq m$. Hence we have the first and second desired inequality. For any $x\in\partial_\Phi(X_\mathrm{int}(n))$, we have $x\in X$. Moreover, there exists a $y\in X_\mathrm{int}(n)^c$ such that $d(x,y)\leq\mathfrak{r}$ holds. We then observe that
\[
d(x,X^c)\leq d(x,y)+ d(y,X^c)\leq n+\mathfrak{r}.
\]
Hence we have $x\in\partial X(n+\mathfrak{r})$ and the third desired inequality holds. For any $y\in\partial_\Phi(X(n+\mathfrak{r}))$, we have $d(y,X)\leq n+\mathfrak{r}$. Moreover, there exists a $z\in X(n+\mathfrak{r})^c$ such that $d(y,z)\leq\mathfrak{r}$ holds. We then observe that
\[
d(y,X)\geq d(z,X)-d(y,z)\geq n.
\]
Hence we have $y\in X^c$ and the last desired inequality holds.
\end{proof}

With the notations above, we define
\begin{align*}
&H_{R,n} :=\sum_{\substack{Y\in\mathcal{P}_0(\Gamma)\\ Y\cap X_\mathrm{int}(n+\mathfrak{r})\neq\emptyset}}\Phi(Y), 
\quad H'_{R,n} :=H_{R,n} -\langle\Omega,H_{R,n}\Omega\rangle I,\\
&H_{B,n} :=\sum_{Y\subseteq \partial X(n+\mathfrak{r})}\Phi(Y), 
\quad H'_{B,n} :=H_{B,n} -\langle\Omega,H_{B,n}\Omega\rangle I.
\end{align*}
We then observe that
\begin{equation}\label{ineq_HB}
\Vert H'_{B,n}\Vert \leq 2\Vert H_{B,n}\Vert 
\leq 2\mathfrak{j}|\partial X(n+\mathfrak{r})|,
\end{equation}
and
\[
H_{R,n} + H_{B,n} = H_{X(n+\mathfrak{r})},\quad \langle\Omega,H'_{R,n}\Omega\rangle =\langle\Omega,H'_{B,n}\Omega\rangle =0, 
\]
where $H_{X(n+\mathfrak{r})}$ is the local Hamiltonian on $X(n+\mathfrak{r})$ constructed out of $\Phi$. 
Moreover, we define
\begin{align*}
&H'_{X(n+\mathfrak{r})}:= H'_{R,n}+H'_{B,n}, \\
&H'_{L,n}:=H_\omega -H'_{X(n+\mathfrak{r})}
\quad\text{with}\quad\mathrm{Dom}(H'_{L,n}):=\mathrm{Dom}(H_\omega).
\end{align*}
By their construction, we have
\[
\langle\Omega,H'_{X(n+\mathfrak{r})}\Omega\rangle 
=\langle\Omega,H'_{L,n}\Omega\rangle =0.
\]
 
\begin{lemma}\label{lem_norm}
We have
\begin{align*}
\Vert[H_\Lambda,H_{R,n}]\Vert
&\leq 2\mathfrak{j}'\mathfrak{j}''|\partial_\Phi (X_\mathrm{int}(n+\mathfrak{r}))|
\quad\text{for any}\quad \Lambda\supseteq X_\mathrm{int}(n+\mathfrak{r}),\\
\Vert[H_\Lambda,H_{X(n+\mathfrak{r})}] \Vert
&\leq 2\mathfrak{j}'\mathfrak{j}''|\partial_\Phi (X(n+\mathfrak{r}))|
\quad\quad\text{for any}\quad \Lambda\supseteq X(n+\mathfrak{r}).
\end{align*}
\end{lemma}
\begin{proof}
Let $\Lambda\supseteq X_\mathrm{int}(n+\mathfrak{r})$ be arbitrary chosen and fixed. Since $H_\Lambda-H_{R,n}\in\mathcal{A}_{X_\mathrm{int}(n+\mathfrak{r})^c}$ and by assumption (A1), we have
\[
[H_\Lambda,H_{R,n}]
=[H_\Lambda-H_{R,n},H_{R,n}]
=\sum_{\substack{Y\subseteq\Lambda\\ Y\cap X_\mathrm{int}(n+\mathfrak{r})\neq\emptyset\\ Y\cap X_\mathrm{int}(n+\mathfrak{r})^c\neq\emptyset}}[H_\Lambda-H_{R,n},\Phi(Y)].
\]
Following \cite[(5.8)]{Hamza-Michalakis-Nachtergaele-Sims}, we obtain that
\[
\Vert[H_\Lambda,H_{R,n}]\Vert
\leq \sum_{y\in\partial_\Phi (X_\mathrm{int}(n+\mathfrak{r}))} \sum_{\substack{Y\in\mathcal{P}_0(\Gamma)\\ Y\ni y}} \sum_{z\in Y} \sum_{\substack{Z\in\mathcal{P}_0(\Gamma)\\ Z\ni z}} \Vert[\Phi(Z),\Phi(Y)]\Vert.
\]
Thus we have
\[
\Vert[H_\Lambda,H_{R,n}]\Vert
\leq 2\mathfrak{j}'\sum_{y\in\partial_\Phi (X_\mathrm{int}(n+\mathfrak{r}))}\sum_{\substack{Y\in\mathcal{P}_0(\Gamma)\\ Y\ni y}}|Y|\Vert\Phi(Y)\Vert
\leq 2\mathfrak{j}'\mathfrak{j}''|\partial_\Phi (X_\mathrm{int}(n+\mathfrak{r}))|.
\]
Hence the first inequality has been shown.

On the other hand, we also have
\[
\left[H_\Lambda,H_{X(n+\mathfrak{r})}\right]= \left[H_\Lambda-H_{X(n+\mathfrak{r})}, H_{X(n+\mathfrak{r})}\right]
=\sum_{\substack{Y\in\Lambda\\ Y\cap X(n+\mathfrak{r})\neq\emptyset\\ Y\cap X(n+\mathfrak{r})^c\neq\emptyset}}[\Phi(Y),H_{X(n+\mathfrak{r})}].
\]
Hence the above argument also works for showing the second inequality.
\end{proof}

\medskip
For each $\alpha>0$ we define
\[
f_\alpha(t):=\sqrt{\frac{\alpha}{\pi}} e^{-\alpha t^2},\quad t\in\mathbb{R}.
\]
Then for any $c\geq0$, we observe that
\begin{equation}\label{eq_f-alpha}
\int_c^\infty tf_\alpha(t)\,dt
=\frac{1}{2}\sqrt{\frac{\alpha}{\pi}}\int_{c^2}^\infty e^{-\alpha s}\,ds
=\frac{f_\alpha(c)}{2\alpha}
\end{equation}
with substitution $s=t^2$. Then, following \cite[(4.1)-(4.3)]{Hamza-Michalakis-Nachtergaele-Sims}, we obtain the next lemma. 
\begin{lemma}\label{lem_convolusion}
We have
\[
\left\Vert\int_{-\infty}^\infty e^{itH_\omega} f_\alpha(t)\,dt-P_0\right\Vert\leq e^{-\gamma^2/4\alpha}.
\]
\end{lemma}

Note that $\alpha^t$ naturally extends to the whole $\mathcal{B}(\mathcal{H}_\omega)$ by $\alpha^t=\mathrm{Ad}e^{itH_\omega}$. Then we define
\begin{align*}
(A)_\alpha 
&:= \int_{-\infty}^\infty\alpha^t(A)f_{\alpha}(t)\,dt,
\quad A\in\mathcal{B}(\mathcal{H}_\omega),\\
\left(H'_{L,n}\right)_\alpha 
&:= H_\omega-\left(H'_{X(n+\mathfrak{r})}\right)_\alpha
\quad\text{with}\quad\mathrm{Dom}\left(\left(H'_{L,n}\right)_\alpha\right):=\mathrm{Dom}(H_\omega)
\end{align*}
for each $\alpha>0$.

\begin{lemma}\label{lem_alpha}
For any $A\in\mathcal{A}_\mathrm{loc}$ with $\langle\Omega,A\Omega\rangle=0$, we have
\[
\Vert(A)_\alpha\Omega\Vert
\leq \frac{e^{-\gamma^2/4\alpha}}{\gamma}\left\Vert\delta(A)\Omega\right\Vert,
\]
where $\delta$ is defined in Proposition \ref{prop_derivation}(i).
\end{lemma}
\begin{proof}
Since $\langle\Omega,(A)_\alpha\Omega\rangle =\langle\Omega,A\Omega\rangle =0$, we have
\[
\Vert H_\omega(A)_\alpha\Omega\Vert^2
= \int_0^\infty \lambda^2\Vert E(d\lambda)(A)_\alpha\Omega\Vert^2 
\geq \gamma^2 \Vert(I-P_0)(A)_\alpha\Omega\Vert^2
= \gamma^2 \Vert(A)_\alpha\Omega\Vert^2.
\]
By Corollary \ref{cor_dom}, we have
\[
H_\omega(A)_\alpha\Omega
=\int_{-\infty}^\infty e^{itH_\omega}[H_\omega,A]e^{-itH_\omega}\Omega\,f_\alpha(t)\,dt 
= -i(\delta(A))_\alpha\Omega,
\]
where we note that $\delta(A)\in\mathcal{A}\subseteq\mathcal{B}(\mathcal{H}_\omega)$. Hence, we observe that
\begin{equation}\label{eq_alpha1}
\Vert(A)_\alpha\Omega\Vert
\leq \frac{1}{\gamma}\Vert H_\omega(A)_\alpha\Omega\Vert 
= \frac{1}{\gamma} \Vert(\delta(A))_\alpha\Omega\Vert. 
\end{equation}

Moreover, by Corollary \ref{cor_dom} again, we have $\langle\Omega,\delta(A)\Omega\rangle= \langle\Omega,i[H_\omega,A]\Omega\rangle=0$. Thus we have
\[
(\delta(A))_\alpha\Omega 
=\int_{-\infty}^\infty e^{itH_\omega}\delta(A)\Omega f_\alpha(t)\,dt
=\left(\int_{-\infty}^\infty e^{itH_\omega} f_\alpha(t)\,dt-P_0\right)\delta(A)\Omega.
\]
By Lemma \ref{lem_convolusion}, we obtain that
\begin{equation}\label{ineq_alpha2}
\Vert(\delta(A))_\alpha\Omega\Vert 
\leq\left\Vert\int_{-\infty}^\infty e^{itH_\omega} f_\alpha(t)\,dt-P_0\right\Vert \Vert\delta(A)\Omega\Vert
\leq e^{-\gamma^2/4\alpha} \Vert \delta(A)\Omega\Vert.
\end{equation}
The desired assertion follows from equations \eqref{eq_alpha1}, \eqref{ineq_alpha2}.
\end{proof}

\begin{corollary}\label{cor_alpha}
We have
\[
\left\Vert\left(H'_{R,n}\right)_\alpha\Omega\right\Vert 
\leq \frac{2\mathfrak{j}'\mathfrak{j}''}{\gamma}e^{-\gamma^2/4\alpha} |\partial_\Phi (X_\mathrm{int}(n+\mathfrak{r}))|,\ 
\left\Vert\left(H'_{L,n}\right)_\alpha\Omega\right\Vert 
\leq \frac{2\mathfrak{j}'\mathfrak{j}''}{\gamma}e^{-\gamma^2/4\alpha} |\partial_\Phi(X(n+\mathfrak{r}))|.
\]
\end{corollary}

\begin{proof}
The first desired inequality follows by Lemmas \ref{lem_norm}, \ref{lem_alpha}. Since we have  $\Vert(H'_{L,n})_\alpha\Omega\Vert = \Vert(H'_{X(n+\mathfrak{r})})_\alpha \Omega\Vert$, the above argument also works for showing the second desired inequality.
\end{proof}

\medskip
\begin{proposition}\label{prop_MR}
For each $\ell>\mathfrak{r}$ and each $\alpha>0$ there exists a self-adjoint element $M_{R,\ell,\alpha}\in\mathcal{A}_X$ such that
\begin{align*}
&\left\Vert\left(H'_{R,\ell}\right)_\alpha-M_{R,\ell,\alpha}\right\Vert \\
&\quad\leq \frac{2\mathfrak{j}'\mathfrak{j}''c_\mu e^{\mu\mathfrak{r}}}{\sqrt{\pi\alpha}} |\partial_\Phi (X_\mathrm{int}(\ell+\mathfrak{r}))|\, |\partial_\Phi (X_\mathrm{int}(\ell-\mathfrak{r}))| e^{-\mu\ell/2} 
+
\frac{4\mathfrak{j}'\mathfrak{j}''}{\sqrt{\pi\alpha}} |\partial_\Phi (X_\mathrm{int}(\ell+\mathfrak{r}))| e^{-\alpha\ell^2/4v_\mu^2}.
\end{align*}
When $\alpha=1/\ell$, this estimate turns out to be
\[
\left\Vert\left(H'_{R,\ell}\right)_{1/\ell}-M_{R,\ell,1/\ell}\right\Vert
\leq C_1|\partial X(\ell+2\mathfrak{r})|^2\,e^{-C_2\ell}
\]
with some constants $C_1,C_2>0$ independent of the choice of $(X,\ell)$, where we used Lemma \ref{lem_set}.
\end{proposition}
\begin{proof}
Since $(H'_{R,\ell})_\alpha$ is self-adjoint, Proposition \ref{prop_partial_trace} enables us to find a self-adjoint element $M_{R,\ell,\alpha}\in\mathcal{A}_X$ such that 
\[ 
\left\Vert\left(H'_{R,\ell}\right)_\alpha-M_{R,\ell,\alpha}\right\Vert \leq\sup_{\substack{B\in\mathcal{A}_X'\cap\mathcal{A}\\\Vert B\Vert\leq1}} \Vert[(H'_{R,\ell})_\alpha, B]\Vert.
\]

Let $B\in\mathcal{A}_X'\cap\mathcal{A}$ with $\Vert B\Vert\leq1$ be arbitrarily chosen and fixed. For any $t\in\mathbb{R}$ and for any $\Lambda\in \mathcal{P}_0(\Gamma)$ with $X_\mathrm{int}(\ell+\mathfrak{r}) \subseteq\Lambda$, we have
\[
\frac{d}{dt} \left[\alpha_\Lambda^t\left(H'_{R,\ell}\right),B\right]
= \left[\alpha_\Lambda^t \left(i\left[H_\Lambda, H'_{R,\ell}\right]\right) ,B\right]
= \left[\alpha_\Lambda^t \left(i\left[H_\Lambda, H_{R,\ell}\right]\right) ,B\right], 
\]
and hence
\[
\left\Vert\left[\alpha_\Lambda^t\left(H'_{R,\ell}\right),B\right]\right\Vert 
= \left\Vert\int_0^t \left[\alpha_\Lambda^s\left(i\left[H_\Lambda,H_{R,\ell}\right]\right),B\right]\,ds\right\Vert
\leq \int_{\min\{0,t\}}^{\max\{0,t\}} \left\Vert\left[\alpha_\Lambda^s\left(\left[H_\Lambda,H_{R,\ell}\right]\right),B\right]\right\Vert\,ds.
\]

By assumption (A1) and since $\ell>\mathfrak{r}$, we have 
\[
[H_\Lambda,H_{R,\ell}]
= \sum_{\substack{Y\in\mathcal{P}_0(\Gamma) \\Y\cap X_\mathrm{int}(\ell+\mathfrak{r}) \neq\emptyset}} \sum_{\substack{Z\subseteq\Lambda\\ Z\cap Y\neq\emptyset}} [\Phi(Z),\Phi(Y)] \in \mathcal{A}_{X_\mathrm{int}(\ell-\mathfrak{r})}.
\]
In fact, we observe that $Y\subseteq X_\mathrm{int}(\ell)$ and $Z\subseteq X_\mathrm{int}(\ell-\mathfrak{r})$ for every $Y$ and every $Z$ in the above sum. Thus each summand $[\Phi(Y),\Phi(Z)]$ in the above sum belong to $\mathcal{A}_{X_\mathrm{int}(\ell-\mathfrak{r})}$. Since $B\in \mathcal{A}_X' \cap \mathcal{A}$ and $X_\mathrm{int}(\ell-\mathfrak{r}) \subseteq X$, Corollary \ref{cor_Lieb-Robinson} and Lemma \ref{lem_norm} show that 
\begin{align*}
\left\Vert\left[\alpha_\Lambda^t\left(\left[H_\Lambda, H_{R,\ell}\right]\right),B\right]\right\Vert
&\leq c_\mu\Vert[H_\Lambda,H_{R,\ell}]\Vert\, |\partial_\Phi (X_\mathrm{int}(\ell-\mathfrak{r}))| \,e^{-\mu(d(X_\mathrm{int}(\ell-\mathfrak{r}),X^c)-v_\mu|t|)} \\
&\leq 2\mathfrak{j}'\mathfrak{j}''c_\mu |\partial_\Phi (X_\mathrm{int} (\ell+\mathfrak{r}))|\, |\partial_\Phi (X_\mathrm{int}(\ell-\mathfrak{r}))| e^{\mu\mathfrak{r}}e^{-\mu\ell/2}
\end{align*}
for any $t\in [-T,T]$ with $T=\ell/2v_\mu$ and for any $\Lambda\supseteq X_\mathrm{int}(\ell-\mathfrak{r})$. Hence, we have
\[
\left\Vert\left[\alpha_\Lambda^t\left(H'_{R,\ell}\right),B\right]\right\Vert \leq 2\mathfrak{j}'\mathfrak{j}''c_\mu |\partial_\Phi (X_\mathrm{int}(\ell+\mathfrak{r}))|\, |\partial_\Phi (X_\mathrm{int}(\ell-\mathfrak{r}))| e^{\mu\mathfrak{r}}e^{-\mu\ell/2}\,|t|
\]
and taking $\Lambda\nearrow\Gamma$, we have
\[
\left\Vert\left[\alpha^t\left(H'_{R,\ell}\right),B\right]\right\Vert
\leq 2\mathfrak{j}'\mathfrak{j}''c_\mu|\partial_\Phi (X_\mathrm{int}(\ell+\mathfrak{r}))| \,|\partial_\Phi (X_\mathrm{int}(\ell-\mathfrak{r}))| e^{\mu\mathfrak{r}}e^{-\mu\ell/2}\,|t|
\]
for any $t\in [-T,T]$.

On the other hand, by Lemma \ref{lem_norm}, we have
\[
\left\Vert\left[\alpha_\Lambda^t \left(\left[H_\Lambda,H_{R,\ell}\right]\right),B\right]\right\Vert 
\leq 2\Vert[H_\Lambda,H_{R,\ell}]\Vert
\leq 4\mathfrak{j}'\mathfrak{j}'' |\partial_\Phi (X_\mathrm{int}(\ell+\mathfrak{r}))|
\]
for any $t\in\mathbb{R}$ and for any $\Lambda\supseteq X_\mathrm{int}(\ell+\mathfrak{r})$. Hence, we have
\[
\left\Vert\left[\alpha_\Lambda^t\left(H'_{R,\ell}\right),B\right]\right\Vert
\leq 4\mathfrak{j}'\mathfrak{j}''|\partial_\Phi (X_\mathrm{int}(\ell+\mathfrak{r}))|\,|t|
\]
and taking $\Lambda\nearrow\Gamma$, we have
\[
\left\Vert\left[\alpha^t\left(H'_{R,\ell}\right),B\right]\right\Vert 
\leq 4\mathfrak{j}'\mathfrak{j}''|\partial_\Phi (X_\mathrm{int}(\ell+\mathfrak{r}))|\,|t|
\]
for any $t\in\mathbb{R}$.

Therefore, we have 
\begin{align*}
\left\Vert\left[\alpha^t\left(H'_{R,\ell}\right),B\right]\right\Vert
&\leq 
\begin{cases}
2\mathfrak{j}'\mathfrak{j}''c_\mu|\partial_\Phi (X_\mathrm{int}(\ell+\mathfrak{r}))| \,|\partial_\Phi (X_\mathrm{int}(\ell-\mathfrak{r}))| e^{\mu\mathfrak{r}} e^{-\mu\ell/2}\,|t| & \text{if $|t|<T$,} \\
4\mathfrak{j}'\mathfrak{j}'' |\partial_\Phi (X_\mathrm{int}(\ell+\mathfrak{r}))|\,|t| & \text{if $|t|\geq T$,}
\end{cases}
\end{align*}
and
\begin{align*}
\left\Vert\left[\left(H'_{R,\ell}\right)_\alpha,B\right]\right\Vert
&\leq \int_{-\infty}^\infty \left\Vert\left[\alpha^t\left(H'_{R,\ell}\right),B\right]\right\Vert f_\alpha(t)\,dt \\
&\leq 4\mathfrak{j}'\mathfrak{j}'' c_\mu|\partial_\Phi (X_\mathrm{int}(\ell+\mathfrak{r}))|\,|\partial_\Phi (X_\mathrm{int}(\ell-\mathfrak{r}))| e^{\mu\mathfrak{r}}e^{-\mu\ell/2} \left(\frac{1}{2\sqrt{\pi\alpha}}\right)\\
&\quad+ 8\mathfrak{j}'\mathfrak{j}'' |\partial_\Phi (X_\mathrm{int}(\ell+\mathfrak{r}))| \left(\frac{e^{-\alpha T^2}}{2\sqrt{\pi\alpha}}\right),
\end{align*}
by equation \eqref{eq_f-alpha}, for any $B\in \mathcal{A}_X'\cap\mathcal{A}$ with $\Vert B\Vert\leq1$. Hence we are done. 
\end{proof}

\begin{corollary}\label{cor_PR} 
For each $\ell>\mathfrak{r}$, let $E_{R,\ell}$ be the spectral measure of $M_{R,\ell,1/\ell}$ and define 
\[
O_R(X,\ell)
:= E_{R,\ell}\big([-\xi(\ell),\xi(\ell)]\big)\in\mathcal{A}_X\quad 
\text{with}
\quad\xi(\ell)
:=\max\big\{ e^{-\gamma^2\ell/16}, e^{-\mu\ell/4}, e^{-\ell/16 v_\mu^2}\big\}.
\]
Then one can find two positive constants $C_1,C_2>0$ (independent of the choice of $(X,\ell)$) such that
\[
\big\Vert\big(O_R(X,\ell)-I\big)\Omega\big\Vert
\leq C_1|\partial X(\ell+2\mathfrak{r})|^2\,e^{-C_2\ell}.
\]
\end{corollary}
\begin{proof}
By the Chebyshev inequality,
\[
\Vert M_{R,\ell,1/\ell}\Omega\Vert^2 
\geq\xi(\ell)^2 \big\Vert E_{R,\ell}\big([-\xi(\ell),\xi(\ell)]^c\big)\Omega\big\Vert^2 
=\xi(\ell)^2\big\Vert\big(O_R(X,\ell)-I\big)\Omega\big\Vert^2.
\]
Thus we have
\[
\Vert(O_R(X,\ell)-I)\Omega\Vert 
\leq \frac{1}{\xi(\ell)}\left\{\left\Vert\left(H'_{R,\ell}\right)_{1/\ell}\Omega\right\Vert +\left\Vert\left(H'_{R,\ell}\right)_{1/\ell}-M_{R,\ell,1/\ell}\right\Vert\right\}.
\]
The desired assertion follows by Corollary \ref{cor_alpha}, Proposition \ref{prop_MR} and Lemma \ref{lem_set}.
\end{proof}

\medskip
By Lemma \ref{lem_norm}, we have
\begin{align*}
\Vert[H_\Lambda,H_{B,\ell}]\Vert
&= \Vert[H_\Lambda,H_{X(\ell+\mathfrak{r})}-H_{R,\ell}]\Vert
\leq \Vert[H_\Lambda,H_{X(\ell+\mathfrak{r})}]\Vert+\Vert[H_\Lambda,H_{R,\ell}]\Vert\\
&\leq 2\mathfrak{j}'\mathfrak{j}''\big(|\partial_\Phi (X(\ell+\mathfrak{r}))|+ |\partial_\Phi(X_\mathrm{int}(\ell+\mathfrak{r}))|\big)
\end{align*}
By inequality \eqref{ineq_HB}, we also have 
\[
\left\Vert\left(H'_{B,\ell}\right)_\alpha\right\Vert 
\leq \Vert H'_{B,\ell}\Vert
\leq 2\mathfrak{j}|\partial X(\ell+\mathfrak{r})|.
\]
The same proof as Proposition \ref{prop_MR} with $H_{B,\ell}\in \mathcal{A}_{\partial X(\ell+\mathfrak{r})}$ shows the following: 

\begin{proposition}\label{prop_MB}
For each $\ell>\mathfrak{r}$ and each $\alpha>0$, there exists a self-adjoint element $M_{B,\ell,\alpha}\in \mathcal{A}_{\partial X(2\ell+\mathfrak{r})}$ such that $\Vert M_{B,\ell,\alpha}\Vert \leq 2\mathfrak{j}|\partial X(\ell+\mathfrak{r})|$ and
\begin{align*}
\left\Vert\left(H'_{B,\ell}\right)_\alpha-M_{B,\ell,\alpha}\right\Vert 
&\leq \frac{2\mathfrak{j}'\mathfrak{j}''c_\mu e^{\mu\mathfrak{r}}}{\sqrt{\pi\alpha}} \big(|\partial_\Phi (X(\ell+\mathfrak{r}))|+ |\partial_\Phi(X_\mathrm{int}(\ell+\mathfrak{r}))|\big) |\partial_\Phi (\partial X(\ell+2\mathfrak{r}))|\,e^{-\mu\ell/2} \\
&\quad+ \frac{4\mathfrak{j}'\mathfrak{j}''}{\sqrt{\pi\alpha}} \big(|\partial_\Phi (X(\ell+\mathfrak{r}))|+ |\partial_\Phi(X_\mathrm{int}(\ell+\mathfrak{r}))|\big) e^{-\alpha\ell^2/4v_\mu^2}.
\end{align*}
When $\alpha=1/\ell$, this estimate turns out to be
\[
\left\Vert\left(H'_{B,\ell}\right)_{1/\ell}-M_{B,\ell,1/\ell}\right\Vert \leq C_1|\partial X(\ell+2\mathfrak{r})|^2\, e^{-C_2\ell}
\]
with some constants $C_1,C_2>0$ independent of the choice of $(X,\ell)$, where we used Lemma \ref{lem_set}.
\end{proposition}

\medskip
Thanks to assumption (A2) in subsection \ref{S2.3}, we can choose an (irreducible) unitary representation $U:G\to\mathcal{A}_X\cong M_{d_X}(\mathbb{C})$ of a finite group $G$ so that $\{U(g)\,;\,g\in G\}''= \mathcal{A}_X$. ({\it n.b.}, for each $k\geq2$ the Specht (irreducible) representation of the symmetric group of rank $k$ corresponding to the partition $(k-1,1)$ must be of dimension $k-1$.) We then define an operator
\[
M_{L,n.\alpha} 
:= \left(H'_{L,n}\right)_\alpha+\frac{1}{|G|} \sum_{g\in G}U(g)^*\left[\left(H'_{L,n}\right)_\alpha,U(g)\right]
\]
with $\mathrm{Dom}(M_{L,n,\alpha})= \mathrm{Dom}((H'_{L,n})_\alpha)= \mathrm{Dom}(H_\omega)$. This construction is indeed due to Matsui \cite{Matsui13}. Since $U(g)\in\mathcal{A}_X$, i.e., a local observable and since $U(g)\mathrm{Dom}(H_\omega)\subseteq\mathrm{Dom}(H_\omega)$ by Corollary \ref{cor_dom}, we have
\[
U(g)^*\left[\left(H'_{L,n}\right)_\alpha,U(g)\right]\phi
=-\left( iU(g)^*\delta(U(g))+U(g)^*\left[\left(H'_{X(n+\mathfrak{r})}\right)_\alpha,U(g)\right]\right)\phi
\] 
for every $\phi\in\mathrm{Dom}(H_\omega)$. Hence $U(g)^*[(H'_{L,n})_\alpha,U(g)]$ extends to the whole $\mathcal{H}_\omega$ as a bounded self-adjoint observable living in $\mathcal{A}$. In particular, $M_{L,n,\alpha}$ must be self-adjoint. 

\begin{proposition}\label{prop_affiliated}
The self-adjoint operator $M_{L,n,\alpha}$ is affiliated with $\mathcal{A}_X'=(\mathcal{A}_X'\cap \mathcal{A})''$ on $\mathcal{H}_\omega$. In partcular, its spectral measure lives in $(\mathcal{A}_X'\cap\mathcal{A})''$. 
\end{proposition}
\begin{proof}
By spectral theory, we can choose a sequence $A_k=A_k^*\in\mathcal{B}(\mathcal{H}_\omega)$ in such a way that $A_k\phi\to\left(H'_{L,n}\right)_\alpha\phi$ as $k\to\infty$ for every $\phi\in\mathrm{Dom}(H_\omega)$. By Corollary \ref{cor_dom}, we have 
\begin{align*}
M_{L,n,\alpha}\phi
&= \left(H'_{L,n}\right)_\alpha\phi 
+\frac{1}{|G|}\sum_{g\in G}U(g)^* \left[\left(H'_{L,n}\right)_\alpha, U(g)\right]\phi\\
&= \frac{1}{|G|}\sum_{g\in G} U(g)^*\left(H'_{L,n}\right)_\alpha U(g)\phi= \lim_{k\to\infty}\frac{1}{|G|}\sum_{g\in G} U(g)^*A_k U(g)\phi
\end{align*}
for every $\phi\in\mathrm{Dom}(H_\omega)$. 
Since all the $(1/|G|)\sum_{g\in G} U(g)^*A_k U(g)$ fall into $\{U(g)\,;\,g\in G\}'=\mathcal{A}_X'$, we obtain that 
\begin{align*} 
M_{L,n,\alpha} V\phi 
&=\lim_{k\to\infty}\frac{1}{|G|}\sum_{g\in G}U(g)^*A_k U(g)V\phi\\
&=V\lim_{k\to\infty}\frac{1}{|G|}\sum_{g\in G}U(g)^*A_k U(g)\phi =V M_{L,n,\alpha}\phi
\end{align*}
for every $(V,\phi)\in U(\mathcal{A}_X)\times\mathrm{Dom}(H_\omega)$, where we used $V\phi, U(g)\phi\in\mathrm{Dom}(H_\omega)$ by Corollary \ref{cor_dom}. The desired assertion follows by Lemma \ref{lem_commutant}. 
\end{proof}

The next proposition may be regarded as a generalization (or a reproduction) of Matsui's claim \cite[equation (4.6)]{Matsui13} and the most important step in the proof. 

\begin{proposition}\label{prop_ML}
With $n:=\ell>\mathfrak{r}$, we have
\begin{align*}
&\sup_{\substack{\phi\in\mathrm{Dom}(H_\omega)\\\Vert\phi\Vert=1}} \left\Vert\left(\left(H'_{L,\ell}\right)_\alpha-M_{L,\ell,\alpha}\right)\phi\right\Vert 
\leq 
\frac{8\mathfrak{j}'\mathfrak{j}''}{\sqrt{\pi\alpha}}|\partial_\Phi(X(\ell+\mathfrak{r}))| e^{-\alpha\ell^2/4v_\mu^2}\\
&\qquad\qquad+\frac{2c_\mu\mathfrak{j}'\mathfrak{j}'' e^{\mu\mathfrak{r}}}{\sqrt{\pi\alpha}} |\partial_\Phi(X(\ell+\mathfrak{r}))| \left(|\partial_\Phi X|+\frac{\mathfrak{j}'e^{\mu\mathfrak{r}}}{\mu v_\mu} |\partial_\Phi(X(\ell+2\mathfrak{r}))|\, |\partial_\Phi(X(2\ell+\mathfrak{r}))| \right)e^{-\mu\ell/2}.
\end{align*}
When $\alpha=1/\ell$, this estimate turns out to be
\[
\sup_{\substack{\phi\in\mathrm{Dom}(H_\omega)\\\Vert\phi\Vert=1}} \left\Vert\left( \left(H'_{L,\ell}\right)_{1/\ell}-M_{L,\ell,1/\ell} \right)\phi\right\Vert
\leq C_1|\partial X(2\ell+\mathfrak{r})|^3\,e^{-C_2\ell}
\]
with some constants $C_1,C_2>0$ independent of the choice of $(X,\ell)$, where we used Lemma \ref{lem_set}.
\end{proposition}
\begin{proof}
Let $\phi$ be a unit vector in $\mathrm{Dom}(H_\omega)$ and $\ell>\mathfrak{r}$ be arbitrarily given. Any $Y\in\mathcal{P}_0(\Gamma)$ with $X\cap Y\neq\emptyset$ and $\Phi(Y)\neq0$ must be included in $X(\mathfrak{r})\subseteq X(2\ell+\mathfrak{r})$. By Corollary \ref{cor_dom},
\[
[H_\omega,U(g)]\phi
=\lim_{\Lambda\nearrow\Gamma} [H_\Lambda,U(g)]\phi
=[H_{X(2\ell+\mathfrak{r})},U(g)]\phi
\]
holds for every $g\in G$. Hence we have
\begin{align*}
\left\Vert\left(\left(H'_{L,\ell}\right)_\alpha-M_{L,\ell,\alpha}\right)\phi\right\Vert
&= \left\Vert\frac{1}{|G|}\sum_{g\in G} U(g)^* \left[H_{X(2\ell+\mathfrak{r})}-\left(H'_{X(\ell+\mathfrak{r})}\right)_\alpha,U(g)\right]\phi\right\Vert\\
&\leq \frac{1}{|G|} \sum_{g\in G}\int_{-\infty}^\infty \left\Vert\left[\alpha_{X(2\ell+\mathfrak{r})}^t \left(H_{X(2\ell+\mathfrak{r})}-H'_{X(\ell+\mathfrak{r})}\right),U(g)\right] \right\Vert\,f_\alpha(t)\,dt\\
&\qquad+ 2\int_{-\infty}^\infty \left\Vert \left(\alpha_{X(2\ell+\mathfrak{r})}^t-\alpha^t\right)\left(H'_{X(\ell+\mathfrak{r})}\right)\right\Vert \,f_\alpha(t)\,dt
\qquad=:(i)+(ii).
\end{align*}

\medskip
We first find an explicit bound for the integral $(i)$. For any $t\in\mathbb{R}$ and for any $g\in G$, we have
\begin{align*}
&\frac{d}{dt}\left[\alpha_{X(2\ell+\mathfrak{r})}^t \left(H_{X(2\ell+\mathfrak{r})}-H'_{X(\ell+\mathfrak{r})}\right),U(g)\right]\\
&\quad= \left[\alpha_{X(2\ell+\mathfrak{r})}^t \left(i\left[H_{X(2\ell+\mathfrak{r})}, H_{X(2\ell+\mathfrak{r})}-H'_{X(\ell+\mathfrak{r})}\right]\right),U(g)\right]\\
&\quad= \left[\alpha_{X(2\ell+\mathfrak{r})}^t \left(i\left[H_{X(2\ell+\mathfrak{r})}, H_{X(\ell+\mathfrak{r})}\right]\right),U(g)\right]
\end{align*}
and hence
\begin{align*}
&\left\Vert\left[\alpha_{X(2\ell+\mathfrak{r})}^t \left(H_{X(2\ell+\mathfrak{r})}-H'_{X(\ell+\mathfrak{r})}\right), U(g)\right]\right\Vert\\
&\quad\leq \int_{\min\{0,t\}}^{\max\{0,t\}} \left\Vert\left[\alpha_{X(2\ell+\mathfrak{r})}^s \left(\left[H_{X(2\ell+\mathfrak{r})}, H_{X(\ell+\mathfrak{r})}\right]\right),U(g)\right]\right\Vert\,ds.
\end{align*}
By assumption (A1), we have
\begin{align*}
\left[H_{X(2\ell+\mathfrak{r})}, H_{X(\ell+\mathfrak{r})}\right]
&=\left[H_{X(2\ell+\mathfrak{r})}, H_{X(2\ell+\mathfrak{r})}-H_{X(\ell+\mathfrak{r})}\right]\\
&=\sum_{\substack{Y\subseteq X(2\ell+\mathfrak{r})\\Y\cap X(\ell+\mathfrak{r})^c\neq\emptyset}} \sum_{\substack{Z\subseteq X(2\ell+\mathfrak{r})\\Z\cap Y\neq\emptyset}}[\Phi(Z),\Phi(Y)] \in\mathcal{A}_{X(\ell-\mathfrak{r})^c}'\cap\mathcal{A}.
\end{align*}
In fact, we observe that $Y\subseteq X(\ell)^c$ and $Z\subseteq X(\ell-\mathfrak{r})^c$ for every $Y$ and every $Z$ in the above sum. Thus each summand $[\Phi(Z),\Phi(Y)]$ in the above sum belong to $\mathcal{A}_{X(\ell-\mathfrak{r})^c}' \cap\mathcal{A}$. Since $U(g)\in\mathcal{A}_X$, $X\subseteq X(\ell-\mathfrak{r})$ and $d(X,X(\ell-\mathfrak{r})^c)\geq \ell-\mathfrak{r}$, Corollary \ref{cor_Lieb-Robinson} and Lemma \ref{lem_norm} show that 
\begin{align*}
\left\Vert\left[\alpha_{X(2\ell+\mathfrak{r})}^t \left(\left[H_{X(2\ell+\mathfrak{r})}, H_{X(\ell+\mathfrak{r})}\right]\right),U(g)\right]\right\Vert
&=\left\Vert\left[\left[H_{X(2\ell+\mathfrak{r})}, H_{X(\ell+\mathfrak{r})}\right], \alpha_{X(2\ell+\mathfrak{r})}^{-t}\big(U(g)\big)\right]\right\Vert\\
&\leq c_\mu\left\Vert\left[H_{X(2\ell+\mathfrak{r})}, H_{X(\ell+\mathfrak{r})}\right]\right\Vert\, |\partial_\Phi X|e^{-\mu(\ell-\mathfrak{r}-v_\mu|t|)} \\
&\leq 2\mathfrak{j}'\mathfrak{j}''c_\mu |\partial_\Phi(X(\ell+\mathfrak{r}))|\,|\partial_\Phi X|e^{\mu\mathfrak{r}}e^{-\mu\ell/2}
\end{align*}
for any $g\in G$ and for any $t\in[-T,T]$ with $T=\ell/2v_\mu$. Hence we have
\[
\left\Vert\left[\alpha_{X(2\ell+\mathfrak{r})}^t \left(\left[H_{X(2\ell+\mathfrak{r})}, H_{X(\ell+\mathfrak{r})}\right]\right),U(g)\right]\right\Vert
\leq 2\mathfrak{j}'\mathfrak{j}''c_\mu |\partial_\Phi(X(\ell+\mathfrak{r}))|\,|\partial_\Phi X|e^{\mu\mathfrak{r}}e^{-\mu\ell/2}|t|
\]
for any $g\in G$ and for any $t\in[-T,T]$.

On the other hand, by Lemma \ref{lem_norm}, we have
\[
\left\Vert\left[\alpha_{X(2\ell+\mathfrak{r})}^t \left(\left[H_{X(2\ell+\mathfrak{r})}, H_{X(\ell+\mathfrak{r})}\right]\right),U(g)\right]\right\Vert
\leq 2\left \Vert\left[H_{X(2\ell+\mathfrak{r})}, H_{X(\ell+\mathfrak{r})}\right]\right\Vert
\leq4\mathfrak{j}'\mathfrak{j}''|\partial_\Phi(X(\ell+\mathfrak{r}))|
\]
and hence 
\[
\left\Vert\left[\alpha_{X(2\ell+\mathfrak{r})}^t \left(H_{X(2\ell+\mathfrak{r})}-H'_{X(\ell+\mathfrak{r})}\right),U(g)\right]\right\Vert
\leq 4\mathfrak{j}'\mathfrak{j}''|\partial_\Phi(X(\ell+\mathfrak{r}))||t|
\]
for any $t\in\mathbb{R}$ and for any $g\in G$.

Consequently, we have for any $g\in G$
\begin{align*}
&\left\Vert\left[\alpha_{X(2\ell+\mathfrak{r})}^t \left(H_{X(2\ell+\mathfrak{r})}-H'_{X(\ell+\mathfrak{r})}\right),U(g)\right]\right\Vert\\
&\quad\leq 
\begin{cases}
2\mathfrak{j}'\mathfrak{j}''c_\mu |\partial_\Phi(X(\ell+\mathfrak{r}))|\,|\partial_\Phi X|e^{\mu\mathfrak{r}}e^{-\mu\ell/2}|t| & \text{if $|t|<T$,} \\
4\mathfrak{j}'\mathfrak{j}''|\partial_\Phi(X(\ell+\mathfrak{r}))||t| & \text{if $|t|\geq T$,}
\end{cases}
\end{align*}
and
\begin{align*}
(i)
\leq 4\mathfrak{j}'\mathfrak{j}''c_\mu |\partial_\Phi(X(\ell+\mathfrak{r}))|\,|\partial_\Phi X|e^{\mu\mathfrak{r}}e^{-\mu\ell/2}\left(\frac{1}{2\sqrt{\pi\alpha}}\right)
+8\mathfrak{j}'\mathfrak{j}''|\partial_\Phi(X(\ell+\mathfrak{r}))|\left(\frac{e^{-\alpha T^2}}{2\sqrt{\pi\alpha}}\right)
\end{align*}
by equation \eqref{eq_f-alpha}.

\medskip
We then find an explicit bound for the integral $(ii)$. By Corollary \ref{cor_dom}, for each $t\in\mathbb{R}$ and each unit vector $\phi\in\mathrm{Dom}(H_\omega)$, we have
\[
\frac{d}{dt} \left(\alpha_{X(2\ell+\mathfrak{r})}^t-\alpha^t\right) \left(H'_{X(\ell+\mathfrak{r})}\right)\phi
= \left(\alpha_{X(2\ell+\mathfrak{r})}^t-\alpha^t\right) \left(i\left[H_{X(2\ell+\mathfrak{r})}, H_{X(\ell+\mathfrak{r})}\right]\right)\phi
\]
and hence
\[
\left\Vert\left(\alpha_{X(2\ell+\mathfrak{r})}^t-\alpha^t\right) \left(H'_{X(\ell+\mathfrak{r})}\right)\phi\right\Vert
\leq \int_{\min\{0,t\}}^{\max\{0,t\}} \left\Vert\left(\alpha_{X(2\ell+\mathfrak{r})}^t-\alpha^t\right) \left(\left[H_{X(2\ell+\mathfrak{r})}, H_{X(\ell+\mathfrak{r})}\right]\right)\right\Vert\,ds. 
\]
Since $\mathrm{Dom}(H_\omega)$ is dense subspace of $\mathcal{H}_\omega$, we then observe that
\[
\left\Vert\left(\alpha_{X(2\ell+\mathfrak{r})}^t-\alpha^t\right) \left(H'_{X(\ell+\mathfrak{r})}\right)\right\Vert
\leq \int_{\min\{0,t\}}^{\max\{0,t\}} \left\Vert\left(\alpha_{X(2\ell+\mathfrak{r})}^t-\alpha^t\right) \left(\left[H_{X(2\ell+\mathfrak{r})}, H_{X(\ell+\mathfrak{r})}\right]\right)\right\Vert\,ds. 
\]
Here we have 
\[
[H_{X(2\ell+\mathfrak{r})}, H_{X(\ell+\mathfrak{r})}]
=\sum_{\substack{Y\subseteq X(2\ell+\mathfrak{r})\\Y\cap X(\ell+\mathfrak{r})\neq\emptyset}}[\Phi(Y),H_{X(\ell+\mathfrak{r})}]
\in\mathcal{A}_{X(\ell+2\mathfrak{r})}
\]
by assumption (A1). In fact, we observe that $Y\subseteq X(\ell+2\mathfrak{r})$ for every $Y$ in the above sum. Thus each summand $[\Phi(Y),H_{X(\ell+\mathfrak{r})}]$ in the above sum belong to $\mathcal{A}_{X(\ell+2\mathfrak{r})}$. By Lemmas \ref{lem_time}, \ref{lem_norm},
\begin{align*}
&\left\Vert\left(\alpha_{X(2\ell+\mathfrak{r})}^t-\alpha^t\right) \left(\left[H_{X(2\ell+\mathfrak{r})}, H_{X(\ell+\mathfrak{r})}\right]\right)\right\Vert\\
&\quad\leq
\frac{c_{\mu}\mathfrak{j}'}{\mu v_\mu}
|\partial_\Phi(X(\ell+2\mathfrak{r}))|\, |\partial_\Phi(X(2\ell+\mathfrak{r}))|\, \left\Vert \left[H_{X(2\ell+\mathfrak{r})}, H_{X(\ell+\mathfrak{r})}\right]\right\Vert e^{-\mu(\ell-2\mathfrak{r})} (e^{\mu v_\mu|t|} -1) \\
&\quad\leq
\frac{2c_{\mu}\mathfrak{j}'^2\mathfrak{j}''e^{2\mu\mathfrak{r}}}{\mu v_\mu} |\partial_\Phi(X(\ell+2\mathfrak{r}))| \,|\partial_\Phi(X(2\ell+\mathfrak{r}))|\, |\partial_\Phi(X(\ell+\mathfrak{r}))| e^{-\mu\ell/2}
\end{align*}
holds for any $t\in [-T,T]$ with $T :=\ell/2 v_\mu$. Therefore, we have
\begin{align*}
&\left\Vert\left(\alpha_{X(2\ell+\mathfrak{r})}^t-\alpha^t\right) \left(H'_{X(\ell+\mathfrak{r})}\right)\right\Vert\\
&\quad\leq \frac{2c_{\mu}\mathfrak{j}'^2\mathfrak{j}''e^{2\mu\mathfrak{r}}}{\mu v_\mu} |\partial_\Phi(X(\ell+2\mathfrak{r}))|\, |\partial_\Phi(X(2\ell+\mathfrak{r}))|\, |\partial_\Phi(X(\ell+\mathfrak{r}))| e^{-\mu\ell/2}|t|
\end{align*}
for any $t\in [-T,T]$.

On the other hand, by Lemma \ref{lem_norm}, we have
\[
\left\Vert\left(\alpha_{X(2\ell+\mathfrak{r})}^t-\alpha^t\right) \left(\left[H_{X(2\ell+\mathfrak{r})}, H_{X(\ell+\mathfrak{r})}\right]\right)\right\Vert
\leq 2 \Vert[H_{X(2\ell+\mathfrak{r})}, H_{X(\ell+\mathfrak{r})}]\Vert \leq 4\mathfrak{j}'\mathfrak{j}'' |\partial_\Phi(X(\ell+\mathfrak{r}))|
\]
and hence
\[
\left\Vert\left(\alpha_{X(2\ell+\mathfrak{r})}^t-\alpha^t\right) \left(H'_{X(\ell+\mathfrak{r})}\right)\right\Vert\leq 4\mathfrak{j}'\mathfrak{j}'' |\partial_\Phi(X(\ell+\mathfrak{r}))||t|.
\]

Consequently, we have
\begin{align*}
&\left\Vert\left(\alpha_{X(2\ell+\mathfrak{r})}^t-\alpha^t\right) \left(H'_{X(\ell+\mathfrak{r})}\right)\right\Vert \\
&\quad\leq 
\begin{dcases}
\frac{2c_{\mu}\mathfrak{j}'^2\mathfrak{j}''e^{2\mu\mathfrak{r}}}{\mu v_\mu}
|\partial_\Phi(X(\ell+2\mathfrak{r}))|\, |\partial_\Phi(X(2\ell+\mathfrak{r}))|\, |\partial_\Phi(X(\ell+\mathfrak{r}))|\, e^{-\mu\ell/2}|t| & \text{if $|t|<T$,} \\
4\mathfrak{j}'\mathfrak{j}'' |\partial_\Phi(X(\ell+\mathfrak{r}))| \,|t| & \text{if $|t|\geq T$,}
\end{dcases}
\end{align*}
and
\begin{align*}
(ii)
&\leq
\frac{4c_{\mu}\mathfrak{j}'^2\mathfrak{j}''e^{2\mu\mathfrak{r}}}{\mu v_\mu}
|\partial_\Phi(X(\ell+2\mathfrak{r}))|\, |\partial_\Phi(X(2\ell+\mathfrak{r}))|\, |\partial_\Phi(X(\ell+\mathfrak{r}))| e^{-\mu\ell/2} \left(\frac{1}{2\sqrt{\pi\alpha}}\right)\\
&\qquad+
8\mathfrak{j}'\mathfrak{j}'' |\partial_\Phi(X(\ell+\mathfrak{r}))|\left(\frac{e^{-\alpha T^2}}{2\sqrt{\pi\alpha}}\right). 
\end{align*}
by equation \eqref{eq_f-alpha}. Hence we are done. 
\end{proof}

In the same way as Corollary \ref{cor_PR}, we have the next corollary. 

\begin{corollary}\label{cor_PL}
For each $\ell>\mathfrak{r}$, let $E_{L,\ell}$ be the spectral measure of $M_{L,\ell,1/\ell}$ and define 
\[
O_L(X,\ell)
:=E_{L,\ell}\big([-\eta(\ell),\eta(\ell)]\big) \quad\text{with}\quad\eta(\ell):=\max\{e^{-\gamma^2\ell/16},e^{-\mu\ell/8},e^{-\ell/16v_\mu^2}\}.
\]
By Lemma \ref{lem_set}, this estimate turns out to be
\[
\big\Vert\big(O_L(X,\ell)-I\big) \Omega\big\Vert \leq C_1|\partial X(2\ell+\mathfrak{r})|^3\,e^{-C_2\ell}
\]
with constants $C_1,C_2>0$ independent of the choice of $(X,\ell)$. Moreover by Proposition \ref{prop_affiliated}, we have
$O_L(X,\ell)\in(\mathcal{A}_X'\cap\mathcal{A})''=\mathcal{A}_X'$.
\end{corollary}

\medskip
Let us fix an arbitrary $\ell>\mathfrak{r}$ in what follows. We write 
\[
M_L=M_{L,\ell,1/\ell},\quad 
M_B=M_{B,\ell,1/\ell},\quad 
M_R=M_{R,\ell,1/\ell},\quad 
O_L=O_L(X,\ell),\quad 
O_R=O_R(X,\ell)
\]
for the ease of notation below. Remark that Propositions \ref{prop_MR}, \ref{prop_MB}, \ref{prop_ML} altogether imply that 
\begin{equation}\label{ineq_altogether}
\sup_{\substack{\phi\in\mathrm{Dom}(H_\omega)\\\Vert\phi\Vert=1}} \left\Vert\big(H_\omega-(M_L+M_B+M_R)\big)\phi\right\Vert
\leq C_1|\partial X(2\ell+\mathfrak{r})|^3 e^{-C_2\ell}
\end{equation}
with some constants $C_1,C_2>0$ independent of the choice of $(X,\ell)$. We also define
\begin{align*}
\mathcal{P} 
&:=\int_{-\infty}^\infty e^{itH_\omega} f_{1/\ell}(t)\,dt\\
\widetilde{\mathcal{P}} 
&:=\int_{-\infty}^\infty e^{it(M_L+M_B+M_R)} f_{1/\ell}(t)\,dt\\
\widehat{\mathcal{P}} 
&:=\int_{-\infty}^\infty e^{it(M_L+M_B+M_R)} e^{-itM_R}e^{-itM_L} f_{1/\ell}(t)\,dt.
\end{align*}
Then we observe that
\begin{equation}\label{ineq_phat}
\left\Vert\widehat{\mathcal{P}}\,O_L O_R-P_0\right\Vert\leq C_1|\partial X(2\ell+\mathfrak{r})|^3 e^{-C_2\ell}
\end{equation}
with some constants $C_1,C_2>0$ independent of the choice of $(X,\ell)$, in the same way as \cite{Hamza-Michalakis-Nachtergaele-Sims}. 

We will give a detailed proof of inequality \eqref{ineq_phat}. Following \cite[(4.5)-(4.8)]{Hamza-Michalakis-Nachtergaele-Sims} and by inequality \eqref{ineq_altogether}, we have
\begin{equation}\label{ineq_proof1}
\begin{aligned}
\left\Vert\widetilde{\mathcal{P}} -\mathcal{P}\right\Vert
&\leq \sup_{\substack{\phi\in\mathrm{Dom}(H_\omega)\\\Vert\phi\Vert=1}} \left\Vert\big(H_\omega-(M_L+M_B+M_R)\big)\phi\right\Vert \times\int_{-\infty}^\infty |t|f_{1/\ell}(t)\,dt\\
&\leq C_3|\partial X(2\ell+\mathfrak{r})|^3 e^{-C_4\ell}
\end{aligned}
\end{equation}
with some constants $C_3,C_4>0$ independent of the choice of $(X,\ell)$. On the other hand, following \cite[(4.12)-(4.14)]{Hamza-Michalakis-Nachtergaele-Sims} and by Corollaries \ref{cor_PR}, \ref{cor_PL}, we observe that
\begin{equation}\label{ineq_proof2}
\Vert P_0\,O_L O_R -P_0\Vert 
\leq \Vert(O_L-I)\Omega\Vert +\Vert(O_R-I)\Omega\Vert
\leq C_5|\partial X(2\ell+\mathfrak{r})|^3 e^{-C_6\ell}
\end{equation}
with some constants $C_5,C_6>0$ independent of the choice of $(X,\ell)$, where we note that $O_RO_L=O_LO_R$ by Proposition \ref{prop_affiliated}. Following \cite[(4.17)-(4.19)]{Hamza-Michalakis-Nachtergaele-Sims} and by the definitions of $O_L$ and $O_R$, we observe that 
\[
\left\Vert e^{itM_L}e^{itM_R}O_LO_R-O_LO_R\right\Vert \leq2(\xi(\ell)+\eta(\ell))|t|. 
\]
Hence one can find some constants $C_7,C_8>0$ independent of the choice of $(X,\ell)$ such that
\begin{equation}\label{ineq_proof3}
\left\Vert\left(\widehat{\mathcal{P}} -\widetilde{\mathcal{P}}\right)O_L O_R\right\Vert
\leq \int_{-\infty}^\infty\Vert e^{-it M_L} e^{-it M_R} O_L O_R-O_L O_R\Vert f_\alpha(t)\,dt
\leq C_7 e^{-C_8\ell}.
\end{equation}
Consequently, by inequalities \eqref{ineq_proof1}, \eqref{ineq_proof2}, \eqref{ineq_proof3} and Lemma \ref{lem_convolusion}, one can find some constants $C_1,C_2>0$ independent of the choice of $(X,\ell)$ that make inequality \eqref{ineq_phat} hold. 

\medskip
Therefore, the rest of the proof of Theorem \ref{thm_main} devotes approximating $\widehat{\mathcal{P}}$ by an element of $\mathcal{A}_{\partial X(3\ell+\mathfrak{r})}$. It will be done by proving the next lemma.

\begin{lemma}\label{lem_OB}
There exists a contraction $O_B=O_B(X,\ell)\in\mathcal{A}_{\partial X(3\ell+\mathfrak{r})}$ such that
\[
\left\Vert\widehat{\mathcal{P}}-O_B\right\Vert\leq C_1|\partial X(2\ell+\mathfrak{r})|^4 e^{-C_2\ell}
\]
with some constants $C_1,C_2>0$ independent of the choice of $(X,\ell)$. 
\end{lemma}
\begin{proof}
By Proposition \ref{prop_partial_trace} and since $\widehat{\mathcal{P}}$ is a contraction, there exists a contraction $O_B\in\mathcal{A}_{\partial X(3\ell+\mathfrak{r})}$ such that
\[
\left\Vert\widehat{\mathcal{P}}-O_B\right\Vert \leq \sup_{\substack{B\in \mathcal{A}_{\partial X(3\ell+\mathfrak{r})}'\cap \mathcal{A}\\ \Vert B\Vert\leq1}} \left\Vert\left[\widehat{\mathcal{P}},B\right]\right\Vert.
\]
It suffices to find an explicit upper bound of the following:
\[
\left\Vert\left[\widehat{\mathcal{P}}, B\right]\right\Vert \leq \int_{-\infty}^\infty \left\Vert\left[ e^{it(M_L+M_B+M_R)}e^{-itM_L}e^{-itM_R},B \right]\right\Vert f_{1/\ell}(t)\,dt,
\]
for any $B\in\mathcal{A}_{\partial X(3\ell+\mathfrak{r})}'\cap\mathcal{A}$ with $\Vert B\Vert\leq1$. Thus we choose and fix such an arbitrary $B$ in what follows.

\medskip
Let us introduce a $1$-parameter automorphism group $\beta^t :=\mathrm{Ad}e^{it(M_L+M_B+M_R)}$ on $\mathcal{B}(\mathcal{H}_\omega)$, and consider
\[
G_B(t):=\left[e^{it(M_L+M_B+M_R)}e^{-itM_L}e^{-itM_R},B\right].
\]
Following \cite[(4.29)-(4.38)]{Hamza-Michalakis-Nachtergaele-Sims}, we observe that
\[
\Vert G_B(t)\phi\Vert 
\leq \int_{\min\{0,t\}}^{\max\{0,t\}} \Vert[\beta^s(M_B),B]\Vert\,ds
\]
for each unit vector $\phi\in\mathrm{Dom}(H_\omega)$. Here, we note that $e^{it(M_L+M_B+M_R)}\mathrm{Dom}(H_\omega)\subseteq\mathrm{Dom}(H_\omega)$ and $e^{itM_L}\mathrm{Dom}(H_\omega) \subseteq\mathrm{Dom}(H_\omega)$ by the Stone theorem and $e^{itM_R}\mathrm{Dom}(H_\omega) \subseteq\mathrm{Dom}(H_\omega)$ and $B\mathrm{Dom}(H_\omega) \subseteq\mathrm{Dom}(H_\omega)$ by Corollary \ref{cor_dom}. Since $\mathrm{Dom}(H_\omega)$ is dense subspace of $\mathcal{H}_\omega$, we have 
\begin{align*}
\Vert G_B(t)\Vert
&\leq \int_{\min\{0,t\}}^{\max\{0,t\}} \Vert[\beta^s(M_B),B]\Vert\,ds\\
&\leq \int_{\min\{0,t\}}^{\max\{0,t\}} \Vert[\beta^s(M_B)-\alpha^s(M_B),B]\Vert \,ds + \int_{\min\{0,t\}}^{\max\{0,t\}} \Vert[\alpha^s(M_B),B]\Vert\,ds \\
&=: (i)+(ii).
\end{align*}

\medskip
We first find an explicit bound for the integral $(i)$. For each unit vector $\phi\in\mathrm{Dom}(H_\omega)$, we have
\[
\frac{d}{dr} \left(\alpha^r\left(\beta^{s-r}(M_B)\right)\right)\phi= \alpha^r\left(i\left[H_\omega-(M_L+M_B+M_R),\beta^{s-r}(M_B)\right]\right)\phi
\]
and hence 
\begin{align*}
\left\Vert\big(\alpha^s(M_B)-\beta^s(M_B)\big)\phi\right\Vert
&\leq 2\Vert M_B\Vert |s| \sup_{\substack{\phi\in\mathrm{Dom}(H_\omega)\\\Vert\phi\Vert=1}} \left\Vert\big(H_\omega-(M_L+M_B+M_R)\big)\phi\right\Vert\\
&\leq C_1|\partial X(2\ell+\mathfrak{r})|^4 e^{-C_2\ell} |s|
\end{align*}
with some constants $C_1, C_2 > 0$ independent of the choice of $(X,\ell)$, by inequality \eqref{ineq_altogether}. Since $\mathrm{Dom}(H_\omega)$ is dense subspace of $\mathcal{H}_\omega$, we have
\[
\Vert\alpha^s(M_B)-\beta^s(M_B)\Vert\leq C_1|\partial X(2\ell+\mathfrak{r})|^4 e^{-C_2\ell}\,|s|, 
\]
and hence for any $t\in\mathbb{R}$
\[
(i)\leq C_1|\partial X(2\ell+\mathfrak{r})|^4 e^{-C_2\ell}\,t^2. 
\]

\medskip
We then find an explicit bound for the integral $(ii)$. We have $M_B\in \mathcal{A}_{\partial X(2\ell+\mathfrak{r})}$, $B\in \mathcal{A}'_{\partial X(3\ell+\mathfrak{r})}\cap\mathcal{A}$ and $\partial X(2\ell+\mathfrak{r})\subseteq \partial X(3\ell+\mathfrak{r})$. Thus by Corollary \ref{cor_Lieb-Robinson}, we have
\[
\Vert[\alpha^s(M_B),B]\Vert \leq c_\mu\Vert M_B\Vert\, |\partial_\Phi (\partial X(2\ell+\mathfrak{r}))|e^{-\mu\ell/2} 
\]
for any $s\in [-T,T]$ with $T :=\ell/2v_\mu$. On the other hand, we observe that 
\[
\Vert[\alpha^s(M_B),B]\Vert\leq 2\Vert M_B\Vert.
\]
By Propositions \ref{prop_MB}, one can find constants $C_3,C_4,C_5>0$ independent of the choice of $(X,\ell)$ 
\begin{align*}
(ii)
&\leq 
\begin{cases}
C_3|\partial X(2\ell+\mathfrak{r})|^2 e^{-C_4\ell}|t| & \text{if $|t|<T$,} \\
C_5|\partial X(2\ell+\mathfrak{r})|\,|t| & \text{if $|t|\geq T$.}
\end{cases}
\end{align*}

\medskip
Summing up the above discussions and by equation \eqref{eq_f-alpha}, we obtain the desired assertion. 
\end{proof}

\section{Area law}\label{S4}
Let $(\mathcal{A},\{\mathcal{A}_X\}_{X\in\mathcal{P}_0(\Gamma)},\Phi)$ be the triple associated with a quantum many-body system over a countable discrete metric space $(\Gamma,d)$. Let $\omega$ be a state on $\mathcal{A}$. Throughout this section, we assume that \emph{the Hastings factorization, i.e., the consequence of Corollary \ref{cor_main}, holds for the state $\omega$}. In addition, assumptions (A2) and (A3) have been assumed to hold after subsections \ref{S2.3} and \ref{S3.1}, respectively, and let assumption (A4) below be assumed to hold. 
\begin{itemize} 
\item[(A4)] $d_\infty:=\sup_{x\in\Gamma} d_x<\infty$ with $d_x:=d_{\{x\}},\ x\in\Gamma$. 
\end{itemize}
These assumptions are fulfilled when the triple $(\mathcal{A},\{\mathcal{A}_X\}_{X \in \mathcal{P}_0(\Gamma)}, \Phi)$ is associated with a quantum spin or a fermion system (i.e., all dimensions of spins are uniformly bounded.).

We use the constants $C_1,C_2$ in the consequence of Corollary \ref{cor_main}. Without loss of generality, we may and do assume $C_1>1$. For simplicity, we introduce the following notations:
\[
P_0:=|\Omega\rangle\langle\Omega|,\quad
\epsilon_X(\ell):= C_1|\partial X(\mathfrak{r})|\exp(-C_2\ell)
\]
for each $X\in\mathcal{P}_0(\Gamma)$ and each $\ell>0$. 

For each $X \in\mathcal{P}_0(\Gamma)$, let $\rho_X$ be the density operator of $\omega|_{\mathcal{A}_X}$. We consider the von Neumann entropy 
\begin{equation}\label{eq_vNe}
s(\omega|_{\mathcal{A}_X}) :=-\mathrm{Tr}_X(\rho_X\log \rho_X),
\end{equation}
where $\mathrm{Tr}_X$ is the non-normalized trace on $\mathcal{A}_X$, that is, $\mathrm{Tr}_X(I)=d_X$ with the identity element $I$ of $\mathcal{A}_X$ ({\it n.b.}, $\mathcal{A}_X$ is $*$-isomorphic to a full matrix algebra $M_{d_X}(\mathbb{C})$ by assumpton (A2)). 

\subsection{Evaluation of entanglement entropy}\label{S4.1}
Let $\mathcal{A}\curvearrowright\mathcal{H}_\omega\ni\Omega$ be the GNS representation associated with the state $\omega$. For any $X\in\mathcal{P}_0(\Gamma)$, as in Lemma \ref{lem_commutant}, one can find a unitary operator $W:\mathcal{H}_\omega\to\mathbb{C}^{d_X}\otimes\mathcal{K}$ with some Hilbert space $\mathcal{K}$. Then we have the Schmidt decomposition of $W\Omega$ in $\mathbb{C}^{d_X}\otimes\mathcal{K}$
\[
W\Omega= \sum_{j}\sqrt{\lambda_j} \Psi_j\otimes\Upsilon_j,\quad
0\leq\lambda_{j+1} \leq\lambda_j \leq\cdots \leq\lambda_1\leq1,\quad \sum_{j}\lambda_j=1,
\]
where $\{\Psi_j\}$ and $\{\Upsilon_j\}$ are orthogonal normalized basis in $\mathbb{C}^{d_X}$ and $\mathcal{K}$, respectively. Then we have 
\[
\rho_X= \sum_j \lambda_j W^* \big(
|\Psi_j\rangle\langle\Psi_j|\otimes I
\big) W.
\]
We define 
\[
\rho'_X:= \sum_j \lambda_j W^* \big( I\otimes|\Upsilon_j\rangle\langle\Upsilon_j| \big)W \in\mathcal{A}_X'. 
\]
Then we have
\begin{equation}\label{eq_rhorho}
\rho_X\rho'_X= \sum_{ij} \lambda_i\lambda_j W^* |\Psi_i\otimes\Upsilon_j\rangle \langle\Psi_i\otimes\Upsilon_j| W,
\end{equation}
which is a trace-class operator of $\mathcal{H}_\omega$. We denote the trace on $\mathcal{H}_\omega$ by $\mathrm{Tr}_{\mathcal{H}_\omega}$ and define 
\begin{equation}\label{eq_fidelity}
p_X :=\mathrm{Tr}_{\mathcal{H}_\omega}(\rho_X\rho'_XP_0) =\langle\Omega,\rho_X\rho'_X\Omega\rangle
=\sum_j \lambda_j^3>0.
\end{equation}

\medskip
The Hastings factorization (see the beginning of this section) leads to the next theorem, which is a generalization of \cite[Lemma 2]{Hastings07} and \cite[Lemma 4.3]{Matsui13}. The proof of \cite[Lemma 4.3]{Matsui13} also works in this setting. However, inequality (4.17) of \cite{Matsui13} seems a flaw, by looking at the case when $x_1=x_2=1/2$, $x_3=x_4=\cdots=0$ and $a_1=1$, $a_2=2/3$, $a_3=a_4=\cdots=0$, a counter example; hence we will use inequality \eqref{ineq_prob} below instead.

\begin{theorem}\label{thm_entropy}
There exist constants $C_3,C_4>0$ (depending only on the graph $(\Gamma,d)$ and $d_\infty$ in assumption (A4)) such that
\[
s(\omega|_{\mathcal{A}_X}) 
\leq C_3|\partial X(\mathfrak{r})|\big(\log|\partial X(\mathfrak{r})|\big)^\nu
+ C_4|\partial X(\mathfrak{r})|\left(\log\frac{1}{p_X}\right)^\nu
\]
holds for any $X\in\mathcal{P}_0(\Gamma)$. 
\end{theorem}
\begin{proof}
We sometimes write $O:=O_B(X,\ell)O_L(X,\ell)O_R(X,\ell)$ below for simplicity. We define
\[
\Tilde{\sigma}(X,\ell):= O\rho_X\rho'_X O^*,\quad
\sigma(X,\ell) :=\frac{\Tilde{\sigma}(X,\ell)}{\mathrm{Tr}_{\mathcal{H}_\omega}(\Tilde{\sigma}(X,\ell))}.
\]
Thanks to the Hastings factorization, i.e., the consequence of Corollary \ref{cor_main} (see the beginning of this section), we have 
\begin{align*}
|\mathrm{Tr}_{\mathcal{H}_\omega}\big(P_0\Tilde{\sigma}(X,\ell)\big)-p_X|
&=\left|\langle\Omega,\big(\Tilde{\sigma}(X,\ell)-\rho_X\rho'_X\big)\Omega\rangle\right|\\
&\leq\left|\langle\Omega,(O-P_0)\rho_X\rho'_XO^*\Omega\rangle\right|
+\left|\langle\Omega,\rho_X\rho'_X\big(O^*-P_0\big)\Omega\rangle\right|\\
&\leq\mathrm{Tr}_{\mathcal{H}_\omega}(\rho_X\rho'_X)\Vert O-P_0\Vert\Vert O^*P_0\Vert
+\mathrm{Tr}_{\mathcal{H}_\omega}(\rho_X\rho'_X)\Vert O^*-P_0\Vert\Vert P_0\Vert\\
&\leq 2\epsilon_X(\ell),
\end{align*}
and hence
\[
0<p_X\leq \mathrm{Tr}_{\mathcal{H}_\omega} \big(P_0\Tilde{\sigma}(X,\ell)\big)+ 2\epsilon_X(\ell).
\]
Moreover, we have
\begin{align*}
\mathrm{Tr}_{\mathcal{H}_\omega} \big((1-P_0) \Tilde{\sigma}(X,\ell)\big)
&\leq\mathrm{Tr}_{\mathcal{H}_\omega}(\rho_X\rho'_X) \Vert(1-P_0)O\Vert\Vert O^*\Vert\\
&\leq\Vert(1-P_0)O\Vert
\leq\Vert O-P_0\Vert+\Vert P_0(P_0-O)\Vert
\leq 2\epsilon_X(\ell).
\end{align*}
Hence we obtain that
\begin{equation}\label{ineq_sigma_1}
\begin{aligned}
1-\langle\Omega, \sigma(X,\ell)\Omega\rangle
&=\frac{\mathrm{Tr}_{\mathcal{H}_\omega} \big((1-P_0) \Tilde{\sigma}(X,\ell)\big)}{\mathrm{Tr}_{\mathcal{H}_\omega} \big((1-P_0) \Tilde{\sigma}(X,\ell)\big) + \mathrm{Tr}_{\mathcal{H}_\omega} \big(P_0 \Tilde{\sigma}(X,\ell)\big)}\\
&\leq \frac{2\epsilon_X(\ell)}{2\epsilon_X(\ell)+ \mathrm{Tr}_{\mathcal{H}_\omega} \big(P_0 \Tilde{\sigma}(X,\ell)\big)}
\leq\frac{2\epsilon_X(\ell)}{p_X},
\end{aligned}
\end{equation}
since the function $x\mapsto x/(x+a)$ is non-decreasing for each $a>0$.

By equation \eqref{eq_rhorho}, we have
\[
\sigma(X,\ell)=\sum_{ij}\mu_{ij} |\Xi_{ij}\rangle\langle\Xi_{ij}|,
\]
where we write
\begin{align*}
\Tilde{\Xi}_{ij}:=OW^* (\Psi_i\otimes\Upsilon_j),\quad
\Xi_{ij}:=\frac{\Tilde{\Xi}_{ij}}{\Vert\Tilde{\Xi}_{ij}\Vert},\quad
\mu_{ij}:=\frac{\lambda_i\lambda_j}{\mathrm{Tr}_{\mathcal{H}_\omega}(\Tilde{\sigma}(X,\ell))}\Vert\Tilde{\Xi}_{ij}\Vert^2.
\end{align*}
Then one can find a matrix $C^{ij}$ with indices $k$ and $l$ such that 
\[
\Xi_{ij}=\sum_{kl}C^{ij}(k,l) W^*\Psi_k\otimes\Upsilon_l.
\]
The rank of matrix $C^{ij}$ is the same as the the Schmidt rank of $W\Xi_{ij}$ with respect to $\mathbb{C}^{d_X}\otimes\mathcal{K}$. Since the Schmidt rank of $O_R(X,\ell)O_L(X,\ell) W^*(\Psi_i\otimes\Upsilon_j)$ is $1$ and $O_B(X,\ell)\in\mathcal{A}_{\partial X(3\ell+\mathfrak{r})}$, the rank of matrix $C^{ij}$ is at most $d_{\partial X(3\ell+\mathfrak{r})}$, where $\partial X(3\ell+\mathfrak{r})$ is defined by equation \eqref{eq_def-set}. Then we have 
\begin{align*}
\langle\Omega,\sigma(X,\ell)\Omega\rangle
&=\sum_{ij}\mu_{ij} |\langle\Xi_{ij},\Omega\rangle|^2
=\sum_{ij}\mu_{ij} |\mathrm{Tr}_{\ell^2(\mathbb{N})} (\Lambda^{1/2}E^{ij}C^{ij})|^2\\
&\leq\sum_{ij}\mu_{ij} \mathrm{Tr}_{\ell^2(\mathbb{N})} (\Lambda^{1/2}E^{ij}\Lambda^{1/2}) \mathrm{Tr}_{\ell^2(\mathbb{N})} \big((C^{ij})^*C^{ij}\big)
=\sum_{ij}\mu_{ij} \mathrm{Tr}_{\ell^2(\mathbb{N})} (\Lambda E^{ij}),
\end{align*}
where $E^{ij}$ is the support projection of $C^{ij}$ and $\Lambda$ is a diagonal matrix with $\Lambda(k,l):=\delta_{kl}\lambda_k$. Thanks to \cite[Lemma 4.4]{Matsui13}, we have 
\[
\mathrm{Tr}(\Lambda E^{ij})
\leq\sum_{k=1}^{d_{\partial X(3\ell+\mathfrak{r})}}\lambda_k,
\]
and hence
\begin{equation}\label{ineq_sigma_2}
\langle\Omega,\sigma(X,\ell)\Omega\rangle
\leq \sum_{k=1}^{d_{\partial X(3\ell+\mathfrak{r})}}\lambda_k.
\end{equation}

Therefore, by inequalities \eqref{ineq_sigma_1}, \eqref{ineq_sigma_2} and assumptions (A3), (A4), we have
\begin{equation}\label{ineq_sigma}
\frac{2\epsilon_X(\ell)}{p_X}
\geq 1-\langle\Omega,\sigma(X,\ell)\Omega\rangle
\geq \sum_{j\geq d_{\partial X(3\ell+\mathfrak{r})}+1} \lambda_j
\geq \sum_{j\geq D(\ell)+1} \lambda_j,
\end{equation}
for any $\ell>0$ with $\epsilon_X(\ell)<1$, where we write $D(\ell):=(d_\infty)^{\kappa|\partial X(\mathfrak{r})|(3\ell)^\nu}$ for simplicity. 

One can find an $m_0\in\mathbb{N}$ such that
\begin{equation}\label{ineq_m0d}
\frac{2\epsilon_X(m_0)}{p_X}<1,\quad \frac{2\epsilon_X(m_0-1)}{p_X}\geq1,
\end{equation}
since we have assumed that $C_1>1$, that is,
\begin{equation}\label{ineq_m0}
m_0\leq1+\frac{1}{C_2}\log\frac{2C_1|\partial X(\mathfrak{r})|}{p_X}.
\end{equation}

Moreover, we have
\[
\sum_{j=1}^{K}-x_j\log x_j
\leq \left(\sum_{j=1}^{K}x_j\right)\log K-\left(\sum_{j=1}^{K}x_j\right)\log\left(\sum_{j=1}^{K}x_j\right),
\]
for any $K\in\mathbb{N}$ and for any $x_1,\,x_2,\,\ldots,\,x_K\in[0,1]$, since the function $f(t)=-t\log t\ (t\in[0,1])$ is convex. Thus we have
\begin{equation}\label{eq1}
\sum_{j=1}^{D(m_0)}-\lambda_j\log\lambda_j
\leq\log(D(m_0))+\frac{1}{e}
=\kappa|\partial X(\mathfrak{r})|(3m_0)^\nu\log(d_\infty)+\frac{1}{e}.
\end{equation}

Let us introduce a probability distribution $q$ on $\mathbb{N}$ in such a way that
\[
q(j)=
\begin{dcases}
\dfrac{1}{D(m_0)} \left(1-\dfrac{2\epsilon_X(m_0)}{p_X}\right) 
\phantom{aaaaaaaaaaaaa} \big(1\leq j\leq D(m_0)\big),\\
\dfrac{2}{p_X}\dfrac{\epsilon_X(m_0+n_0m)-\epsilon_X(m_0+n_0(m+1))}{D(m_0+n_0m)-D(m_0+n_0(m+1))}\\
\phantom{aaaaaaaaaaaaaaaaaaaa}
\big(D(m_0+n_0m)+1\leq j\leq D(m_0+n_0(m+1))\big),
\end{dcases}
\]
where we take an $n_0\in\mathbb{N}$ with $e^{-C_2n_0}<1$. Notice that $\log t\leq t-1$ ($t>0$). Letting $t = x/y$ ($x,y>0$) we have $y-y\log y\leq x-y\log x$ for any $x,y>0$ (and the resulting inequality still holds when $y=0$ because of $0\log0 = 0$). Hence we have 
\begin{equation}\label{ineq_prob}
\sum_j-\lambda_j\log\lambda_j
\leq \sum_j-\lambda_j\log q(j). 
\end{equation}
Moreover, by inequality \eqref{ineq_m0d}, we have
\[
q(j)
\geq \frac{2}{p_X} \frac{(1-e^{-C_2n_0})\epsilon_X(m_0+n_0m)}{D(m_0+n_0m)}
\geq \frac{e^{-C_2(n_0m+1)}(1-e^{-C_2n_0})}{D(m_0+n_0m)}
\]
if $j$ is in the interval corresponding to $m$. Hence for any $m\in\mathbb{N}$, we have
\begin{align*}
\sum_{j=D(m_0+n_0m)+1}^{D(m_0+n_0(m+1))} -\lambda_j\log q(j)
&\leq \left(
\sum_{j=D(m_0+n_0m)+1}^{D(m_0+n_0(m+1))} \lambda_j
\right)
\left(
-\log \frac{e^{-C_2(n_0m+1)}(1-e^{-C_2n_0})}{D(m_0+n_0m)}
\right)\\
&\leq \frac{2\epsilon_X(m_0+n_0m)}{p_X} \left( \log\big(D(m_0+n_0m)\big)+G_1(m)\right)\\
&\leq e^{-C_2n_0m} \left( \kappa|\partial X(\mathfrak{r})|(3m_0+3n_0m)^\nu \log(d_\infty)+G_1(m)\right),
\end{align*}
where we write 
\[
G_1(m):=C_2(n_0m+1)-\log(1-e^{-C_2n_0})
\]
and note that the last inequality holds by inequality \eqref{ineq_m0d}. Thus we have
\begin{equation}\label{eq2}
\begin{aligned}
\sum_{j=D(m_0+n_0m)+1}^{D(m_0+n_0(m+1))} -\lambda_j\log q(j)
&\leq
\left(6^\nu\kappa\log(d_\infty)|\partial X(\mathfrak{r})| \frac{m_0^\nu+(n_0m)^\nu}{2}+G_1(m)\right) e^{-C_2n_0m}\\ 
&\leq \Big(C_3|\partial X(\mathfrak{r})|m_0^\nu+G_2(m)|\partial X(\mathfrak{r})|+G_1(m)\Big) e^{-C_2n_0m},
\end{aligned}
\end{equation}
since the function $f(x)=x^\nu\,(\nu\geq1)$ is convex, where we write 
\[
C_3:=\frac{6^\nu\kappa\log(d_\infty)}{2},\quad
G_2(m):=\frac{(6n_0)^\nu\kappa\log(d_\infty)}{2}m^\nu.
\]

By inequalities \eqref{eq1}, \eqref{eq2}, 
\begin{align*}
&s(\omega|_{\mathcal{A}_X})
=\sum_{j=1}^{\infty}-\lambda_j\log\lambda_j\\
&\quad\leq \kappa|\partial X(\mathfrak{r})|(3m_0)^\nu\log(d_\infty)+\frac{1}{e}
+\sum_{m=0}^\infty \Big(C_3|\partial X(\mathfrak{r})|m_0^\nu+G_2(m)|\partial X(\mathfrak{r})|+G_1(m)\Big)e^{-C_2n_0m}\\
&\quad\leq C'_3|\partial X(\mathfrak{r})|m_0^\nu+ C'_4|\partial X(\mathfrak{r})|+C'_5,
\end{align*}
where we write 
\[
C'_3:= 3^\nu\kappa\log(d_\infty)+C_3\sum_{m=0}^\infty e^{-C_2n_0m},\ 
C'_4:=\sum_{m=0}^\infty G_2(m)e^{-C_2n_0m},\ 
C'_5:=\frac{1}{e}+\sum_{m=0}^\infty G_1(m)e^{-C_2n_0m}.
\]
Here, the above infinite sums converge, since $G_1$ and $G_2$ are polynomials and $n_0$ was taken so that $e^{-C_2n_0}<1$. By inequality \eqref{ineq_m0}, we have
\begin{align*}
s(\omega|_{\mathcal{A}_X}) 
&\leq C'_3|\partial X(\mathfrak{r})|\left(1+\frac{1}{C_2}\log\frac{2C_1|\partial X(\mathfrak{r})|}{p_X}\right)^\nu+C'_4|\partial X(\mathfrak{r})|+C'_5\\
&\leq \frac{3^{\nu-1}}{C_2^\nu}
C'_3|\partial X(\mathfrak{r})|
\left(
(C_2+\log2C_1)^\nu
+\left(\log\frac{1}{p_X}\right)^\nu
+(\log|\partial X(\mathfrak{r})|)^\nu
\right)
+C'_4|\partial X(\mathfrak{r})|+C'_5
\end{align*}
since the function $f(x)=x^\nu\,(\nu\geq1)$ is convex. With letting
\[
C_4:=\frac{3^{\nu-1}}{C_2^\nu}C'_3
\quad\text{and}\quad
C_3:=C_4+C_4(C_2+\log2C_1)^\nu+C'_4+C'_5,
\]
we are done.
\end{proof}

\subsection{$1$-dimensional area law}
In this subsection, we obtain a generalization of \cite[Theorem 1]{Hastings07} and \cite[Proposition 4.2]{Matsui13}. The idea of the proof of \cite[Proposition 4.2]{Matsui13} basically works for more general $1$-dimensional quantum many-body systems too and we find that an area law holds for $1$-dimensional quantum many-body systems beyond quantum spin chains (see Theorem \ref{thm_arealaw} below). Here, we split the discussion of \cite[Proposition 4.2]{Matsui13} into Lemma \ref{lem_division} and Proposition \ref{prop_division_interval}. 

Throughout this subsection, we further assume that \emph{the state $\omega$ is pure}. (We remind the reader of the other assumption given at the beginning of this section.) Here we will investigate the lower bound of $p_X$ in the \emph{cylindrical} setting, following Hastings's idea used in the proof of \cite[Proposition 4.2]{Matsui13}.

Let $X,Y \in\mathcal{P}_0(\Gamma)$ satisfy $X\subseteq Y$. By assumption (A2), we have
\[
\mathcal{A}_Y\cong M_{d_Y}(\mathbb{C})= \mathcal{B}(\mathbb{C}^{d_Y}),\quad
\mathcal{A}_X\cong M_{d_X}(\mathbb{C})= \mathcal{B}(\mathbb{C}^{d_X}).
\]
As in Lemma \ref{lem_commutant}, one can find a unitary transform $W:\mathbb{C}^{d_Y}\to\mathbb{C}^{d_X} \otimes\mathcal{K}$ with some Hilbert space $\mathcal{K}$. Then we have 
\[
W(\mathcal{A}_Y\cap\mathcal{A}'_X)W^*= (W(\mathcal{A}_Y\cap\mathcal{A}'_X)W^*)''= \mathcal{B}(\mathcal{K})\cong M_{d_Y/d_X}(\mathbb{C}).
\]
Hence we will use the following identification:
\[
\mathcal{A}_Y\cap\mathcal{A}'_X \cong\mathcal{B}(\mathcal{K}).
\]

Let $\rho_{Y,X}'$ be the density operator of $\omega|_{\mathcal{A}_Y\cap\mathcal{A}'_X}$ via the identification just above. We consider the von Neumann entropy 
\[
s(\omega|_{\mathcal{A}_Y\cap \mathcal{A}'_X}):=\mathrm{Tr}_\mathcal{K} (-\rho_{Y,X}'\log\rho_{Y,X}').
\]
\begin{lemma}\label{lem_division}
If $p_X\leq1/4$ holds and $\ell>0$ is sufficiently large to make $24\sqrt{\epsilon_X(\ell)}<1$ hold, then for any $Y\in\mathcal{P}_0(\Gamma)$ with $Y\supseteq \partial X(3\ell+\mathfrak{r})$, we have
\[
s(\omega|_{\mathcal{A}_Y})
\leq s(\omega|_{\mathcal{A}_{Y\cap X}})+ s(\omega|_{\mathcal{A}_Y\cap \mathcal{A}'_X})
-\frac{1}{2}\log\frac{1}{p_X+6\sqrt{\epsilon_X(\ell)}}
+\log2.
\]
\end{lemma}
\begin{proof}
We write 
\[
O_B:=O_B(X,\ell),\ O_R:=O_R(X,\ell),\ O_L:=O_L(X,\ell),
\]
for simplicity. Thanks to the Hastings factorization (see the beginning of this section) and since 
\[
O_RO_L(O_RO_LO_B)=O_R(O_RO_LO_B)=O_L(O_RO_LO_B)=O_RO_LO_B
\]
holds, we have
\begin{equation}\label{ineq_P_1}
\Vert O_RO_LP_0-P_0\Vert
\leq \Vert O_RO_L(P_0-O_RO_LO_B)\Vert+\Vert O_RO_LO_B-P_0\Vert
\leq 2\epsilon_X(\ell)
\end{equation}
and
\begin{equation}\label{ineq_P_2}
\Vert O_RP_0-P_0\Vert\leq 2\epsilon_X(\ell),\quad
\Vert O_LP_0-P_0\Vert\leq 2\epsilon_X(\ell).
\end{equation}

By inequalities \eqref{ineq_P_1}, we obtain that
\begin{align*}
1-\omega(O_B)
&\leq \Vert P_0-O_BP_0\Vert
\leq \Vert(P_0-O_BO_LO_R)P_0\Vert +\Vert O_B(O_RO_LP_0-P_0)\Vert\\
&\leq \epsilon_X(\ell)+2\epsilon_X(\ell) = 3\epsilon_X(\ell).
\end{align*}
Since $24\sqrt{\epsilon_X(\ell)}<1$ holds, we obtain that
\begin{equation}\label{ineq_fidelity_1}
\omega(O_B)\geq1-3\epsilon_X(\ell) \geq\frac{1}{2}.
\end{equation}

We also write 
\[
\omega_\otimes:=\omega|_{\mathcal{A}_X}\otimes \omega|_{\mathcal{A}\cap\mathcal{A}'_X},\ 
x=x_{X,\ell}:=\omega_\otimes(O_B),\ 
y=y_{X,\ell}:=\omega_\otimes(O_LO_R)=\omega(O_L)\omega(O_R)
\]
for simplicity. Then, we have
\begin{align*}
|\omega_\otimes(O_BO_LO_R)-xy|^2
&=\big|\omega_\otimes
\big((O_B-x)(O_LO_R-y)\big)\big|^2\\
&\leq \omega_\otimes \big((O_B-x)^2\big)\ 
\omega_\otimes \big((O_LO_R-y)^2\big)
\leq \left(x-x^2\right)\left(y-y^2\right)
\end{align*}
by the Cauchy-Schwarz inequality. Furthermore, thanks to the Hastings factorization (see the beginning of this section) again,
\begin{align*}
|p_X-\omega_\otimes(O_RO_LO_B)|
&=\left|\mathrm{Tr}_{\mathcal{H}_\omega} \big(\rho_X\rho'_X(P_0-O_RO_LO_B)\big)\right|\\
&\leq \mathrm{Tr}_{\mathcal{H}_\omega}(\rho_X\rho'_X) \Vert P_0-O_RO_LO_B\Vert
\leq \epsilon_X(\ell).
\end{align*}
Hence we have 
\[
|p_X-xy|
\leq \epsilon_X(\ell)+ \sqrt{x-x^2}\sqrt{y-y^2}.
\]
Here, we observe that
\begin{align*}
1-y
&\leq|\omega(O_L)|\,|\omega(O_R)-1|+|\omega(O_L)-1|\\
&\leq \Vert O_RP_0-P_0\Vert+\Vert O_LP_0-P_0\Vert
\leq 4\epsilon_X(\ell),
\end{align*}
by inequalities \eqref{ineq_P_2}. Therefore, we have
\begin{align*}
x
&\leq \frac{1}{y}\left(p_X+\epsilon_X(\ell) +\sqrt{x-x^2}\sqrt{y-y^2}\right)
\leq \frac{1}{y} \left(p_X+\epsilon_X(\ell)+\frac{1}{2}\sqrt{1-y}\right)\\
&\leq \frac{p_X+\epsilon_X(\ell)+ \sqrt{\epsilon_X(\ell)}}{1-4\epsilon_X(\ell)}
\leq p_X +\frac{5\epsilon_X(\ell)+ \sqrt{\epsilon_X(\ell)}}{1-4\epsilon_X(\ell)}.
\end{align*}
Since $p_X\leq1/4$ and $24\sqrt{\epsilon_X(\ell)}\leq1$, we have
\begin{equation}\label{ineq_fidelity_2}
\omega_\otimes(O_B)
\leq p_X+2\left(5\epsilon_X(\ell)+\sqrt{\epsilon_X(\ell)}\right)
\leq p_X+6\sqrt{\epsilon_X(\ell)}
\leq\frac{1}{2}.
\end{equation}

Thanks to \cite[Proof of Proposition 4.3.9]{Bhatia:book}, we obtain that
\[
s\big(\omega|_{\mathcal{A}_Y}, \omega|_{\mathcal{A}_{Y\cap X}} \otimes\omega|_{\mathcal{A}_Y\cap \mathcal{A}'_X}\big)
= -s(\omega|_{\mathcal{A}_Y})+ s(\omega|_{\mathcal{A}_{Y\cap X}})+ s(\omega|_{\mathcal{A}_Y\cap \mathcal{A}'_X}).
\]
To estimate the relative entropy $s(\omega|_{\mathcal{A}_Y}, \omega|_{\mathcal{A}_{Y\cap X}} \otimes\omega|_{\mathcal{A}_Y\cap\mathcal{A}'_X})$, we define 
\[
\Phi_{O_B}
\begin{pmatrix}
  a_{11} & a_{12} \\
  a_{21} & a_{22}
\end{pmatrix}
:= a_{11}O_B+a_{22}(1-O_B),
\]
which becomes a Schwarz map because
\[
\Phi_{O_B}(A^*A)-\Phi_{O_B}(A)^*\Phi_{O_B}(A)
= |a_{21}|^2O_B+|a_{12}|^2(1-O_B)+ \big(|a_{11}|-|a_{22}|\big)^2O_B(1-O_B)\geq0
\]
for any 
$A=
\begin{pmatrix}
  a_{11} & a_{12} \\
  a_{21} & a_{22}
\end{pmatrix}
\in M_2(\mathbb{C})$. Thanks to the data processing inequality for relative entropy \cite[Theorem 1.5]{Ohya-Petz:book}, we have
\begin{align*}
s\big( \omega|_{\mathcal{A}_Y}, \omega|_{\mathcal{A}_{Y\cap X}} \otimes\omega|_{\mathcal{A}_Y\cap\mathcal{A}'_X} \big)
&\geq s\big( \omega|_{\mathcal{A}_Y}\circ\Phi_{O_B}, \omega|_{\mathcal{A}_{Y\cap X}} \otimes\omega|_{\mathcal{A}_Y\cap\mathcal{A}'_X}\circ\Phi_{O_B} \big)\\
&= \omega(O_B)\log
\left(
\frac{\omega(O_B)}{\omega_\otimes(O_B)}
\right)
+ \big(1-\omega(O_B)\big)\log
\left(
\frac{1-\omega(O_B)}{1-\omega_\otimes(O_B)}
\right),
\end{align*}
since 
\[
\begin{pmatrix}
  \omega(O_B) & 0 \\
  0 & 1-\omega(O_B)
\end{pmatrix}
\quad\text{and}\quad
\begin{pmatrix}
\omega_\otimes(O_B) & 0 \\
0 & 1-\omega_\otimes(O_B)
\end{pmatrix}
\]
are the density operators of $\omega|_{\mathcal{A}_Y}\circ\Phi_{O_B}$ and $\omega|_{\mathcal{A}_{Y\cap X}} \otimes\omega|_{\mathcal{A}_Y\cap\mathcal{A}'_X}\circ\Phi_{O_B}$, respectively.

For any $0\leq s\leq t\leq1$, we write
\[
f(t,s):=t\log\frac{t}{s}+(1-t)\log\frac{(1-t)}{(1-s)}.
\]
Since $f(s,t)$ is a monotone increasing function in $t\in[s,1]$ and since
\[
\omega_\otimes(O_B)\leq \frac{1}{2}\leq \omega(O_B)
\]
holds by inequalities \eqref{ineq_fidelity_1}, \eqref{ineq_fidelity_2}, we observe that
\begin{align*}
s\big(
\omega|_{\mathcal{A}_Y}, \omega|_{\mathcal{A}_{Y\cap X}} \otimes\omega|_{\mathcal{A}_Y\cap\mathcal{A}'_X}
\big)
&\geq \frac{1}{2}\log\frac{1}{2\omega_\otimes(O_B)}
+\frac{1}{2}\log\frac{1}{2(1-\omega_\otimes(O_B))}\\
&\geq \frac{1}{2}\log\frac{1}{\omega_\otimes(O_B)}-\log2
\geq \frac{1}{2}\log\frac{1}{p_X+6\sqrt{\epsilon_X(\ell)}}-\log2,
\end{align*}
where we note that the last inequality holds by inequality \eqref{ineq_fidelity_2}. Hence we are done.
\end{proof}

\medskip
Then we will prove the next proposition.
\begin{proposition}\label{prop_arealaw}
Assume further that $(\Gamma,d)$ is $\mathbb{Z}\times[1,n]^{\nu-1} \ ([1,n]:=\{1,2,\ldots,n\})$ with $\ell^1$-distance and that assumption (A5) below holds. (See the beginning of this section for the other assumptions, especially on the state $\omega$.)
\begin{itemize} 
\item[(A5)] For any $X,Y\in\mathcal{P}_0(\Gamma)$ with $X\subseteq Y$, there exists a $*$-isomorphism $\tau:\mathcal{A}_{Y\setminus X}\to\mathcal{A}_Y\cap\mathcal{A}'_X$ such that $\omega\circ\tau=\omega$ on $\mathcal{A}_{Y\setminus X}$. 
\end{itemize}
Then we have
\[
\sup_{_{\substack{a,b\in\mathbb{Z}\\ a<b}}} s(\omega|_{\mathcal{A}_{S_{a,b}}})<\infty,
\]
where we write $S_{a,b}:=[a,b]\times[1,n]^{\nu-1}.$
\end{proposition}

The next theorem immediately follows from Proposition \ref{prop_arealaw}, a generalization of \cite[Theorem 1]{Hastings07} or \cite[Proposition 4.2]{Matsui13}. 

\begin{theorem}[$1$-dimensional area law for quantum many-body systems]\label{thm_arealaw}
Under the assumption of Theorem \ref{thm_main} with $n=1$, we have $\sup_{m\in\mathbb{N}} s(\omega|_{\mathcal{A}_{[0,m]}})<\infty$.
\end{theorem}

Under assumption (A5), we have 
\[
s(\omega|_{\mathcal{A}_{Y\setminus X}})= s(\omega|_{\mathcal{A}_Y\cap\mathcal{A}'_X}).
\]
In fact, we have $\mathrm{Tr}_{Y\setminus X}(I)= d_{Y\setminus X}=\dim\mathcal{K}= \mathrm{Tr}_\mathcal{K}(\tau(I))$ with the identity element $I$ of $\mathcal{A}_{Y\setminus X}$, where we used the same notation $\mathrm{Tr}_{Y\setminus X}$ as equation \eqref{eq_vNe}. By the uniqueness of the trace on $\mathcal{A}_{Y\setminus X}\cong M_{d_{Y\setminus X}}(\mathbb{C})$, we obtain $\mathrm{Tr}_{Y\setminus X}= \mathrm{Tr}_\mathcal{K}\circ\tau$. Let $\rho$ be the density operator of $\omega|_{\mathcal{A}_{Y\setminus X}}$. For any $B\in\mathcal{A}_Y\cap \mathcal{A}'_X$, we have 
\[
\omega(B)=\omega\big(\tau^{-1}(B)\big)= \mathrm{Tr}_{Y\setminus X}\big(\rho\tau^{-1}(B)\big) =\mathrm{Tr}_{Y\setminus X}\circ \tau^{-1}\big(\tau(\rho)B\big) =\mathrm{Tr}_\mathcal{K} \big(\tau(\rho)B\big). 
\]
Thus $\tau(\rho)$ is the density operator of $\omega|_{\mathcal{A}_Y\cap\mathcal{A}'_X}$. Hence we have
\[
s(\omega|_{\mathcal{A}\cap\mathcal{A}'_X}) 
=\mathrm{Tr}_\mathcal{K}\big(\tau(-\rho\log\rho)\big)
=\mathrm{Tr}_{Y\setminus X}(-\rho\log\rho)
=s(\omega|_{\mathcal{A}_{Y\setminus X}}).
\]

By assumptions (A2), (A5), we have $\mathcal{A}_Y\cong \mathcal{A}_X\otimes(\mathcal{A}_Y\cap\mathcal{A}'_X)\cong\mathcal{A}_X\otimes\mathcal{A}_{Y\setminus X}$ and hence $d_Y=d_Xd_{Y\setminus X}$ for any $X,Y\in\mathcal{P}_0(\Gamma)$ with $X\subseteq Y$. Therefore, we obtain that $d_X=\prod_{x\in X}d_x \leq (d_\infty)^{|X|}$ and 
\begin{equation}\label{ineq_entropy}
s(\omega|_{\mathcal{A}_X}) \leq|X|\log(d_\infty)
\end{equation}
holds for any $X\in\mathcal{P}_0(\Gamma)$. This fact will be used below.

\medskip
\begin{example}\label{ex_spin-A5}
Assumption (A5) is fulfilled, when the triple $(\mathcal{A},\{\mathcal{A}_X\}_{X\in\mathcal{P}_0(\Gamma)},\Phi)$ is associated with a quantum spin system. In fact, since $\mathcal{A}_{Y\setminus X}= \mathcal{A}_Y\cap\mathcal{A}'_X$, the desired $\tau$ can be chosen to be the identity map on $\mathcal{A}_{Y\setminus X}$.
\end{example}
\begin{example}\label{ex_fermion-A5}
Assumption (A5) is also fulfilled, when the triple $(\mathcal{A},\{\mathcal{A}_X\}_{X\in\mathcal{P}_0(\Gamma)},\Phi)$ is associated with a fermion system and the state $\omega$ is parity invariant. In fact, thanks to e.g., \cite[Theorem 4.17, Theorem 4.7 and Corollary 4.8]{Araki-Moriya}, one can find $V_X\in\mathcal{A}_X\cap\mathcal{A}^+$ such that 
\[
\mathcal{A}_Y\cap\mathcal{A}'_X= (\mathcal{A}_{Y\setminus X}\cap \mathcal{A}^+) +V_X(\mathcal{A}_{Y\setminus X} \cap\mathcal{A}^-),
\]
where $\mathcal{A}^+$ and $\mathcal{A}^-$ are the sets of all the elements of even parity or of odd parity, respectively. Hence we can define a $*$-isomorphism $\tau$ by
\[
\tau:\mathcal{A}_{Y\setminus X}\ni A\to A^++V_XA^-\in\mathcal{A}_Y\cap\mathcal{A}'_X 
\]
with a unique decomposition 
\[
A=A^++A^-\in (\mathcal{A}_{Y\setminus X}\cap \mathcal{A}^+)+(\mathcal{A}_{Y\setminus X} \cap\mathcal{A}^-).
\] 
Since $\omega$ is parity invariant, we have $\omega(A^-)=\omega(V_XA^-)=0$. Thus we have $\omega\circ\tau=\omega$ on $\mathcal{A}_{Y\setminus X}$. Hence $\tau$ is a desired one.
\end{example}

Therefore, Proposition \ref{prop_arealaw} as well as Theorem \ref{thm_arealaw} above are applicable to quantum spin and fermion systems of infinite volume successfully.

\medskip
The rest of this section is devoted to proving Proposition \ref{prop_arealaw}. Lemma \ref{lem_division} enables us to prove the next proposition.

\begin{proposition}\label{prop_division_interval}
Under the same assumptions of Proposition \ref{prop_arealaw}, one can find an $\ell_0(n)>\mathfrak{r}$ that makes the following hold: For any $a,b\in\mathbb{Z}$ with $a<b$, there exist $a_0\in[a-\ell_0(n),a]$ and $b_0\in[b,b+\ell_0(n)]$ such that
\[
p_{S_{a_0,b_0}}\geq 4\mathfrak{r}n^{\nu-1}C_1e^{-C_2\ell_0(n)} =:\epsilon_n(\ell_0(n)).
\]
\end{proposition}
\begin{proof}
On the contrary, suppose that for any $\ell>\mathfrak{r}$, one can find $a_\ell,b_\ell\in\mathbb{Z}$ with $a_\ell<b_\ell$ in such a way that 
\[
p_{S_{a,b}}\leq \epsilon_n(\ell)
\]
for any $a\in[a_\ell-\ell,a_\ell]$ and any $b\in[b_\ell,b_\ell+\ell]$. 

Let $m\in\mathbb{N}$ be arbitrarily given and $\ell_0>\mathfrak{r}$ sufficiently large in such a way that $24\sqrt{\epsilon_n(\ell_0)}<1$. Then one can find $a_m,b_m\in\mathbb{Z}$ with $a_m<b_m$ such that
\[
p_{S_{a,b}}\leq \epsilon_n(2^m3\ell_0)\leq\frac{1}{4}
\]
holds for any $a\in[a_m-2^m3\ell_0,a_m]$ and any $b\in[b_m,b_m+2^m3\ell_0]$. 

Assume that $k\in\mathbb{N}$ satisfies $1\leq k\leq m$ and $I_k,J_k\subseteq\mathbb{Z}$ are intervals satisfying
\begin{equation}\label{property_IkJk}
|I_k|=|J_k|=2^k3\ell_0,\ 
I_k\subseteq[a_m-2^m3\ell_0,a_m],\ 
J_k\subseteq[b_m,b_m+2^m3\ell_0].
\end{equation}
Then we can find an $S_{a,b}$ such that 
\[
I_k=[a-2^{k-1}3\ell_0,a+2^{k-1}3\ell_0-1]
\quad\text{and}\quad
J_k=[b-2^{k-1}3\ell_0+1,b+2^{k-1}3\ell_0]
\]
hold. We have $p_{S_{a,b}}\leq \epsilon_n(2^m3\ell_0)$ due to the choice of $a_m$ and $b_m$, and 
\begin{align*}
&(I_k\cup J_k)\times[1,n]^{\nu-1}
\supseteq \partial S_{a,b}(3(2^{k-1}-1)\ell_0 +\mathfrak{r})\\
&\qquad= \{x\in S_{a,b}\,;\, d(x,S_{a,b}^c)\leq 3(2^{k-1}-1)\ell_0+\mathfrak{r}\} 
\cup\{y\in S_{a,b}^c\,;\, d(y,S_{a,b})\leq 3(2^{k-1}-1)\ell_0+\mathfrak{r}\},
\end{align*}
see equation \eqref{eq_def-set} for the notation. By assumption (A5) and Lemma \ref{lem_division} with $\ell=(2^{k-1}-1)\ell_0$, we have 
\begin{align*}
&s\left(\omega|_{\mathcal{A}_{(I_k\cup J_k)\times[1,n]^{\nu-1}}}\right)\\
&\quad\leq s\left(\omega|_{\mathcal{A}_{\left((I_k\cup J_k)\times[1,n]^{\nu-1}\right)\cap S_{a,b}}}\right)+ s\left(\omega|_{\mathcal{A}_{\left((I_k\cup J_k)\times[1,n]^{\nu-1}\right)\setminus S_{a,b}}}\right)\\
&\qquad-\frac{1}{2}\log\frac{1}{p_{s_{a,b}}+6\sqrt{\epsilon_n((2^{k-1}-1)\ell_0)}}+\log2\\
&\quad\leq s\left(\omega|_{\mathcal{A}_{\left((I_k\cup J_k)\times[1,n]^{\nu-1}\right)\cap S_{a,b}}}\right)+ s\left(\omega|_{\mathcal{A}_{\left((I_k\cup J_k)\times[1,n]^{\nu-1}\right)\setminus S_{a,b}}}\right)
-\frac{C_2(2^{k-1}-1)\ell_0}{4}+C_5,
\end{align*}
where we note that $24\sqrt{\epsilon_n((2^{k-1}-1)\ell_0)}<1$ and write $C_5:=\log(7^24\mathfrak{r}n^{\nu-1}C_1)/4+\log2$.

For any $k\in\mathbb{N}$ with $1\leq k\leq m$, one can find two pairs $(I_{k-1},J_{k-1})$ and $(I'_{k-1},J'_{k-1})$ satisfying property \eqref{property_IkJk}, respectively, such that
\[
I_{k-1}\cup J_{k-1}=I_k\cup J_k\cap [a,b],\quad
I'_{k-1}\cup J'_{k-1} =(I_k\cup J_k)\setminus [a,b]. 
\]
Hence we have 
\begin{equation}\label{ineq_division}
\begin{aligned}
&s\left(\omega|_{\mathcal{A}_{(I_k\cup J_k)\times[1,n]^{\nu-1}}}\right)\\
&\quad\leq
s\left(\omega|_{\mathcal{A}_{(I_{k-1}\cup J_{k-1})\times[1,n]^{\nu-1}}}\right)+
s\left(\omega|_{\mathcal{A}_{(I'_{k-1}\cup J'_{k-1})\times[1,n]^{\nu-1}}}\right)
-\frac{C_2(2^{k-1}-1)\ell_0}{4}+C_5.
\end{aligned}
\end{equation}
Here, by inequality \eqref{ineq_entropy}, we observe that
\[
s\left(\omega|_{\mathcal{A}_{(I_0\cup J_0)\times[1,n]^{\nu-1}}}\right) \leq 6\ell_0n^{\nu-1}\log(d_\infty)
\]
for any pair $(I_0,J_0)$ satisfying property \eqref{property_IkJk}. Using inequality \eqref{ineq_division} repeatedly up to $k=m$, we have
\begin{align*}
s\left(\omega|_{\mathcal{A}_{(I_m\cup J_m)\times[1,n]^{\nu-1}}}\right)
&\leq 2^m6\ell_0n^{\nu-1}\log(d_\infty)+ \sum_{k=1}^{m} 2^{k-1} \left(C_5-\frac{C_2(2^{m-k}-1)\ell_0}{4}\right)\\
&\leq 2^m6\ell_0n^{\nu-1}\log(d_\infty)-m2^{m-3}C_2\ell_0+\frac{C_2\ell_0}{4}(2^m-1)+C_5(2^m-1)\\
&\leq 2^{m-3}\left(\ell_0 \left(2^36n^{\nu-1}\log(d_\infty)-(m-2) C_2\right)+2^3C_5\right).
\end{align*}
If we take an $m\in\mathbb{N}$ and an $\ell_0>0$ large enough in such a way that 
\[
\frac{2^36n^{\nu-1}\log(d_\infty)+1}{C_2}+2<m,\quad 2^3C_5=
2\log(7^24\mathfrak{r}n^{\nu-1}C_1)+8\log2
<\ell_0, 
\]
then we arrive at $s\left(\omega|_{\mathcal{A}_{(I_m\cup J_m)\times[1,n]^{\nu-1}}}\right)<0$, a contradiction.
\end{proof}

\begin{proof}[Proof of Proposition \ref{prop_arealaw}]
Let $a,b\in\mathbb{Z}$ with $a<b$ be arbitrarily given. Thanks to Proposition \ref{prop_division_interval}, we find an $\ell_0(n)>0$ and an interval $[a_0,b_0]\subseteq\mathbb{Z}$ such that 
\[
[a,b]\subseteq[a_0,b_0],\quad
p_{S_{a_0,b_0}}\geq\epsilon_n(\ell_0(n)).
\]

Here, by the data processing inequality for relative entropy \cite[Theorem 1.5]{Ohya-Petz:book},
\[
s(\omega|_{\mathcal{A}_X})=s(\rho_X,I) =s(E_X^Y(\rho_Y),E_X^Y(I)) \leq s(\rho_Y,I)=s(\omega|_{\mathcal{A}_Y})
\]
holds for any $X\subseteq Y$, where $E_X^Y:\mathcal{A}_Y\cong \mathcal{A}_X\otimes (\mathcal{A}'_X\cap\mathcal{A}_Y) \to\mathcal{A}_X$ is the conditional expectation (normalized partial trace) and $s(\,\cdot,\cdot\,)$ denotes the relative entropy, defined as $s(A,B):=\mathrm{Tr}_X(A(\log A-\log B))$ for positive elements $A,B\in\mathcal{A}_X$. Here, we note that $s(A,B)$ is independent the choice of $X\in\mathcal{P}_0(\Gamma)$ with $A,B\in\mathcal{A}_X$, since we identify $A\in\mathcal{A}_X$ with $A\otimes I\in \mathcal{A}_X\otimes\mathcal{A}_{X'}\cong\mathcal{A}_{X\cup X'}$ for any $X,X'\in\mathcal{P}_0(\Gamma)$. Thus we have
\[
s\left(\omega|_{\mathcal{A}_{S_{a,b}}}\right)\leq s\left(\omega|_{\mathcal{A}_{S_{a_0,b_0}}}\right).
\]
Thanks to Theorem \ref{thm_entropy} and due to the choice of $a_0$ and $b_0$, we have
\begin{align*}
s\left(\omega|_{\mathcal{A}_{S_{a_0,b_0}}}\right)
&\leq
4\mathfrak{r}C_3n^{\nu-1}\big(\log(4\mathfrak{r}n^{\nu-1})\big)^\nu
+ 4\mathfrak{r}C_4n^{\nu-1}\left(\log\frac{1}{\epsilon_n(\ell_0(n))}\right)^\nu\\
&\leq 4\mathfrak{r}C_3n^{\nu-1}\big(\log(4\mathfrak{r}n^{\nu-1})\big)^\nu
+ 4\mathfrak{r}C_4n^{\nu-1}(C_2\ell_0(n))^\nu<\infty,
\end{align*}
since we have $|\partial S_{a_0,b_0}(\mathfrak{r})|= 4\mathfrak{r}n^{\nu-1}$ and $\mathbb{Z}\times[1,n]^{\nu-1}$ equipped with $\ell^1$-distance is $\nu$-regular. Hence we are done.
\end{proof}

\appendix 
\section{Lieb-Robinson bound}\label{Appendix_Lieb-Robinson}
In this appendix, we will give a proof of Theorem \ref{thm_Lieb-Robinson} and Corollary \ref{cor_Lieb-Robinson}.
\begin{proof}[Proof of Theorem \ref{thm_Lieb-Robinson}]
We may and do assume that $\mathcal{A}$ acts on a Hilbert space. Let $\Lambda\in\mathcal{P}_0(\Gamma)$ and $B\in\mathcal{A}$ be given. For each $Z\subseteq \Lambda$ and $t\in\mathbb{R}$ we define 
\[
c_B^\Lambda(Z,t):=\sup_{A\in\mathcal{A}_Z\setminus\{0\}}\frac{\Vert[\alpha^t_\Lambda(A),B]\Vert}{\Vert A\Vert}.
\]

For every $A\in\mathcal{A}_Z$ with $Z\subseteq\Lambda$ let a function $f:\mathbb{R}\to\mathcal{A}$ be defined by $f(t) := [\alpha^t_\Lambda(\alpha^{-t}_Z(A)),B]$. By assumption (A1) and since $\ell>\mathfrak{r}$, we observe that 
\begin{equation*}
\frac{d}{dt}f(t)=i
\left[
\sum_{\substack{Z'\subseteq\Lambda\\ Z'\cap Z^c\neq\emptyset}}\alpha^t_\Lambda(\Phi(Z')),f(t)
\right]+ i \sum_{\substack{Z'\subseteq\Lambda\\ Z'\cap Z^c\neq\emptyset}}
\Big[
[\alpha^t_\Lambda(\Phi(Z')) ,B],\alpha^t_\Lambda(\alpha^{-t}_Z(A))
\Big].
\end{equation*}
By \cite[Lemma 2.3]{Nachtergaele-Sims-Young19}, we obtain that
\begin{equation*}\label{ineq_f(t)}
\begin{aligned}
\Vert f(t)\Vert 
&\leq 
\Vert f(0)\Vert+ \sum_{\substack{Z'\subseteq\Lambda\\ Z'\cap Z^c\neq\emptyset}} \int_{\min\{t,0\}}^{\max\{t,0\}} \Big\Vert \big[ [\alpha^s_\Lambda(\Phi(Z')),B],\alpha^s_\Lambda( \alpha^{-s}_Z(A)) \big] \Big\Vert\,ds\\
&\leq \Vert f(0)\Vert+2 \Vert A\Vert \sum_{\substack{Z'\subseteq\Lambda\\ Z'\cap Z^c\neq\emptyset}}
\int_{\min\{t,0\}}^{\max\{t,0\}} \Big\Vert [\alpha^s_\Lambda(\Phi(Z')),B] \Big\Vert\,ds.
\end{aligned}
\end{equation*}
It follows that 
\begin{equation}\label{ineq_integral_inequality_1}
c^\Lambda_B(Z,t)
\leq 
c^\Lambda_B(Z,0) 
+2\sum_{\substack{Z'\subseteq\Lambda\\ Z'\cap Z^c\neq\emptyset}}\Vert\Phi(Z')\Vert\int_{\min\{t,0\}}^{\max\{t,0\}}c^\Lambda_B(Z',s)\,ds.
\end{equation}

\medskip
On the other hand, we assume that $X,Y,\Lambda\in \mathcal{P}_0(\Gamma)$ satisfy $X\subseteq Y\cap\Lambda$ and that $A\in\mathcal{A}_X$ and $B\in\mathcal{A}_Y'\cap\mathcal{A}$ hold. Then we have $c^\Lambda_B(X,0)=0$. Thus by inequality \eqref{ineq_integral_inequality_1} with $Z=X$ and $Z=Z'$ in order, we obtain that
\begin{align*}
&c_B^\Lambda(X,t)
\leq 
2\sum_{\substack{Z'\subseteq\Lambda\\ Z'\cap X^c\neq\emptyset}}\Vert\Phi(Z')\Vert\int_{\min\{t,0\}}^{\max\{t,0\}} c_B^\Lambda(Z',s)\,ds \\
&\leq 
2\sum_{\substack{Z_1\subseteq\Lambda\\ Z_1\cap X^c\neq\emptyset}}\Vert\Phi(Z_1)\Vert \int_{\min\{t,0\}}^{\max\{t,0\}} \left( c_B^\Lambda(Z_1,0)+ 2\sum_{\substack{Z_2\subseteq\Lambda\\ Z_2\cap Z_1^c\neq\emptyset}} \int_{\min\{s,0\}}^{\max\{s,0\}} c_B^\Lambda(Z_2,r)\,dr \right)\,ds \\
&\leq 
2\Vert B\Vert\sum_{\substack{Z_1\subseteq\Lambda\\ Z_1\cap X^c\neq\emptyset}} \Vert\Phi(Z_1)\Vert\,\delta_{Y^c}(Z_1)\,2|t|\\
&\qquad+
4\sum_{\substack{Z_1\subseteq\Lambda\\ Z_1\cap X^c\neq\emptyset}}\Vert\Phi(Z_1)\Vert \sum_{\substack{Z_2\subseteq\Lambda\\ Z_2\cap Z_1^c\neq\emptyset}}\Vert\Phi(Z_2)\Vert 
\int_{\min\{t,0\}}^{\max\{t,0\}}\int_{\min\{s,0\}}^{\max\{s,0\}}c_B^\Lambda(Z_2,r)\,ds\,dr.
\end{align*} 
Iterating this procedure we have 
\[
c_B^\Lambda(X,t)
\leq 2\Vert B\Vert\left\{\sum_{n=1}^{N}\frac{a_n}{n!}(2|t|)^n\right\}+ r_{N+1}(t), 
\]
where
\begin{align*}
a_n 
&= 
\sum_{\substack{Z_1\subseteq\Lambda\\ Z_1\cap X^c\neq\emptyset}}\sum_{\substack{Z_2\subseteq\Lambda\\ Z_2\cap Z_1^c\neq\emptyset}} \cdots \sum_{\substack{Z_n\subseteq\Lambda\\ Z_n\cap Z_{n-1}^c\neq\emptyset}}\delta_{Y^c}(Z_n)\,\prod_{k=1}^{n} \Vert\Phi(Z_k)\Vert, \\
r_{N+1}(t) 
&= 
2^{N+1}\sum_{\substack{Z_1\subseteq\Lambda\\ Z_1\cap X^c\neq\emptyset}}\sum_{\substack{Z_2\subseteq\Lambda\\ Z_2\cap Z_1^c\neq\emptyset}}\cdots\sum_{\substack{Z_{N+1}\subseteq\Lambda\\ Z_{N+1}\cap Z_{N}^c\neq\emptyset}} \prod_{k=1}^{N+1} \Vert\Phi(Z_k)\Vert \\
&\quad
\int_{\min\{t,0\}}^{\max\{t,0\}}\int_{\min\{s_1,0\}}^{\max\{s_1,0\}}\cdots\int_{\min\{s_N,0\}}^{\max\{s_N,0\}}c_B^\Lambda(Z_{N+1}, s_{N+1})\,ds_{N+1} \cdots ds_2\,ds_1.
\end{align*}
By a careful counting argument, we have, for any $n\in\mathbb{N}$, 
\[
\sum_{\substack{Z_1\subseteq\Lambda\\ Z_1\cap X^c\neq\emptyset}} \cdots \sum_{\substack{Z_n\subseteq\Lambda\\ Z_n\cap Z_{n-1}^c\neq\emptyset}} \prod_{k=1}^n \Vert\Phi(Z_k)\Vert
\leq \sum_{w_0\in\partial_\Phi X}\sum_{w_1,\ldots,w_n\in\Lambda} \sum_{\substack{Z_1\subseteq\Lambda\\ w_0,w_1\in Z_1}} \cdots \sum_{\substack{Z_n\subseteq\Lambda\\ w_{n-1},w_n\in Z_n}} \prod_{k=1}^n \Vert\Phi(Z_k)\Vert.
\]
With this estimate and the notations $c_F, \Vert\Phi\Vert_F$, we obtain that 
\[
a_n\leq\Vert\Phi\Vert_F^n c_F^{n-1} \sum_{w_0\in\partial_\Phi X}
\sum_{y\in Y^c} F(d(w_0,y))
\]
and 
\begin{align*}
r_{N+1}(t)
&\leq \Vert\Phi\Vert_F^{N+1}c_F^N\sum_{w_0\in \partial_\Phi X}\sum_{w_{N+1}\in\Lambda} F(d(w_0,w_{N+1})) \\
&\leq 
\frac{2\Vert B\Vert|\partial_\Phi X| \Vert F\Vert}{c_F} \frac{(2|t|\Vert\Phi\Vert_F c_F)^{N+1}}{(N+1)!} \to 0 
\end{align*}
as $N\to\infty$ for any $t\in\mathbb{R}$. Therefore, we conclude that
\[
c_B^\Lambda(X,t)
\leq \frac{2\Vert B\Vert}{c_F}\left(e^{2\Vert\Phi\Vert_Fc_F|t|}-1\right) \sum_{x\in\partial_\Phi X} \sum_{y\in Y^c} F(d(x,y)), 
\]
implying the desired assertion.
\end{proof}

\begin{proof}[Proof for Corollary \ref{cor_Lieb-Robinson}]
Let $X,Y,\Lambda \in \mathcal{P}_0(\Gamma)$ with $X \subseteq Y \cap \Lambda$ be arbitrarily given. Assume that $A\in\mathcal{A}_X$ and $B\in\mathcal{A}_Y' \cap \mathcal{A}$. We have a new $F$-function $F_\mu(r):=e^{-\mu r}F(r)$ (see e.g., \cite[Appendix 2]{Nachtergaele-Sims-Young19}). Since $\Phi$ is bounded and has finite range $\mathfrak{r}$, we have $\Vert\Phi\Vert_{F_\mu}\leq e^{\mu\mathfrak{r}}\Vert \Phi\Vert_F<\infty$. Applying Theorem \ref{thm_Lieb-Robinson} with $F=F_\mu$, we have
\begin{align*}
\Vert[\alpha_\Lambda^t(A),B]\Vert 
&\leq\frac{2\Vert A\Vert\Vert B\Vert}{c_{F_\mu}} \left(e^{2\Vert\Phi\Vert_{F_\mu} c_{F_\mu}|t|}-1\right) \sum_{x\in\partial_\Phi X} \sum_{y\in Y^c} e^{-\mu d(x,y)}F(d(x,y))\\
&\leq \frac{2\Vert A\Vert\Vert B\Vert}{c_{F_\mu}} e^{2\Vert\Phi\Vert_{F_\mu} c_{F_\mu}|t|} |\partial_\Phi X|\,e^{-\mu d(\partial_\Phi X,Y^c)} \Vert F\Vert.
\end{align*}
Hence we get the desired inequality with letting $v_\mu:=2\Vert\Phi\Vert_{F_\mu} c_{F_\mu}/\mu$ and $c_\mu:=2\Vert F\Vert/c_{F_\mu}$. 
\end{proof}

\section{Proof of Propositions \ref{prop_time} and \ref{prop_derivation}} \label{Appendix_derivation}
In this appendix, we will give a self-contained proof of Propositions \ref{prop_time} and \ref{prop_derivation}. Throughout this appendix, we keep the setup of subsection \ref{S2.2} and use the same notations as in section \ref{S1} and subsections \ref{S2.1}, \ref{S2.2}. Moreover, we define the \emph{thermodynamic limit} as follows: Let $\lim_{\Lambda\nearrow\Gamma}A_\Lambda=A$ be such that for each monotone non-decreasing sequence $\{\Lambda_n\}$ of elements of $\mathcal{P}_0(\Gamma)$ with $\bigcup_{n=1}^\infty\Lambda_n=\Gamma$, we have
\[
\lim_{n\to\infty}A_{\Lambda_n}=A. 
\]

\begin{proof}[Proof of Proposition \ref{prop_time}]
Let $A\in\mathcal{A}_X$ with $X\in\mathcal{P}_0(\Gamma)$ be arbitrarily given. By assumption (A1), for any $\Lambda,\Lambda' \in\mathcal{P}_0(\Gamma)$ with $X\subseteq\Lambda\subseteq\Lambda'$, we have
\begin{align*}
\left\Vert\left(\alpha_{\Lambda'}^t-\alpha_\Lambda^t\right)(A)\right\Vert
&\leq \left\Vert \int^{\max\{0,t\}}_{\min\{0,t\}} \frac{d}{ds} \left( \alpha_{\Lambda'}^{t-s} \alpha_\Lambda^s(A)\right) \,ds \right\Vert
\leq \int^{\max\{0,t\}}_{\min\{0,t\}} \Vert[H_{\Lambda'}- H_\Lambda, \alpha^s_\Lambda(A)]\Vert\,ds\\
&\leq \int_{\min\{0,t\}}^{\max\{0,t\}} \sum_{\substack{Z\subseteq\Lambda'\\ Z\cap\Lambda\neq\emptyset\\ Z\cap\Lambda^{c}\neq\emptyset}} \Vert[\Phi(Z),\alpha_\Lambda^s(A)]\Vert\,ds.
\end{align*}
For any $Z\in\mathcal{P}_0(\Gamma)$ in the above sum, we observe that $Z^c\supseteq\Lambda_\mathrm{int}(\mathfrak{r}):=\{x\in\Lambda\,;\,d(x,\Lambda^c) >\mathfrak{r}\}$. By assumption (A1), we have $\Phi(Z)\in \mathcal{A}_{Z^c}' \cap\mathcal{A}\subseteq \mathcal{A}_{\Lambda_\mathrm{int}(\mathfrak{r})}' \cap\mathcal{A}$. 

If we assume that $\Lambda\in\mathcal{P}_0(\Gamma)$ is sufficiently large in such a way that $d(X,\Lambda^c)>\mathfrak{r}$, then we observe that $X\subseteq \Lambda_\mathrm{int}(\mathfrak{r})$. Since $A\in\mathcal{A}_X, \Phi(Z)\in \mathcal{A}_{\Lambda_\mathrm{int}(\mathfrak{r})}' \cap\mathcal{A}$ and $X\subseteq\Lambda_\mathrm{int}(\mathfrak{r})\cap\Lambda$, Theorem \ref{thm_Lieb-Robinson} (Lieb-Robinson bound) leads to 
\[
\Vert[\Phi(Z),\alpha_\Lambda^s(A)]\Vert
\leq \frac{2\Vert A\Vert}{c_F} \left(e^{2\Vert\Phi\Vert_F c_F|s|} -1\right) \Vert\Phi(Z)\Vert \sum_{x\in\partial_\Phi X} \sum_{y\in \Lambda_\mathrm{int}(\mathfrak{r})^c} F(d(x,y)).
\]
Hence we have
\begin{align*}
&\left\Vert\left(\alpha_{\Lambda'}^t-\alpha_\Lambda^t\right)(A)\right\Vert\\
&\leq\frac{2\Vert A\Vert}{c_F} \int_{\min\{0,t\}}^{\max\{0,t\}} \big(e^{2\Vert\Phi\Vert_F c_F|s|}-1\big) \,ds \sum_{\substack{Z\subseteq\Lambda'\\ Z\cap\Lambda\neq\emptyset\\ Z\cap\Lambda^{c}\neq\emptyset}} \Vert \Phi(Z)\Vert \sum_{x\in\partial_\Phi X} \sum_{y\in \Lambda_\mathrm{int}(\mathfrak{r})^c} F(d(x,y)).
\end{align*}

Let $\{\Lambda_n\}$ be a monotone non-decreasing sequence of elements of $\mathcal{P}_0(\Gamma)$ with $\bigcup_{n=1}^\infty\Lambda_n=\Gamma$ in what follows. Then one can find a sufficiently large $n_0\in\mathbb{N}$ in such a way that $d(X,\Lambda_{n_0}^c)>\mathfrak{r}$. For any $n,m\in\mathbb{N}$ with $n_0\leq n<m$, we have
\begin{align*}
&\left\Vert\left(\alpha_{\Lambda_m}^t-\alpha_{\Lambda_n}^t\right)(A)\right\Vert\\
&\leq \frac{2\Vert A\Vert}{c_F} \int_{\min\{0,t\}}^{\max\{0,t\}} \big(e^{2\Vert\Phi\Vert_F c_F|s|}-1\big) \,ds \sum_{x\in\partial_\Phi X} \left( \sum_{\substack{Z\subseteq\Lambda_m\\ Z\cap\Lambda_n^c\neq\emptyset}} \Vert\Phi(Z)\Vert \sum_{y\in {\Lambda_n}_\mathrm{int}(\mathfrak{r})^c} F(d(x,y))\right).
\end{align*}
Here, for any $x\in\partial_\Phi X$, we observe that
\begin{align*}
\sum_{\substack{Z\subseteq\Lambda_m\\ Z\cap\Lambda_n^{c}\neq\emptyset}} \Vert\Phi(Z)\Vert \sum_{y\in {\Lambda_n}_\mathrm{int}(\mathfrak{r})^c} F(d(x,y))
&\leq \sum_{z\in\Lambda_m\backslash\Lambda_n} \sum_{\substack{Z\in \mathcal{P}_0(\Gamma)\\ z\in Z}} \Vert\Phi(Z)\Vert \sum_{y\in\Gamma} F(d(x,y))\\
&\leq \sum_{z\in\Lambda_m\backslash\Lambda_n} \sum_{y\in\Gamma} \sum_{\substack{Z\in\mathcal{P}_0(\Gamma)\\ y,z\in Z}} \Vert\Phi(Z)\Vert F(d(x,y))\\
&\leq \Vert\Phi\Vert_F \sum_{z\in\Lambda_m\backslash\Lambda_n} \sum_{y\in\Gamma} F(d(y,z)) F(d(x,y))\\
&\leq \Vert\Phi\Vert_F c_F \sum_{z\in\Lambda_m\backslash \Lambda_n} F(d(x,z)).
\end{align*}
Notice that $\sum_{z\in\Lambda_m\backslash \Lambda_n} F(d(x,z))\to0$ as $n,m\to\infty$, since $\sum_{z\in\Gamma} F(d(x,z))\leq \Vert F\Vert$. Hence, for any $A\in\mathcal{A}_\mathrm{loc}$, $\lim_{n\to\infty}\alpha_{\Lambda_n}^t(A)$ converge uniformly in $t$ on any compact subsets. Moreover the resulting limit can be proved to be independent of the choice of $\{\Lambda_n\}$.

Write $\alpha^t(A):=\lim_{n\to\infty}\alpha_{\Lambda_n}^t(A)$ for any $A\in\mathcal{A}_\mathrm{loc}$, and then $\alpha^t$ is an isometry. Thus we can extend $\alpha^t$ to the whole $\mathcal{A}$. Then $\alpha^t$ is strongly continuous $1$-parameter automorphism group. In fact, $\alpha^{t+s}=\alpha^t\alpha^s$ holds, since $\alpha_\Lambda^{t+s}=\alpha_\Lambda^t\alpha_\Lambda^s$ holds for any $\Lambda\in\mathcal{P}_0(\Gamma)$. On the other hand, for each $A\in\mathcal{A}$, one can find $A_\epsilon\in\mathcal{A}_\mathrm{loc}$ and $n\in\mathbb{N}$ such that $\Vert A-A_\epsilon\Vert\leq \epsilon$ and $\sup_{t\in[-1,1]} \Vert\alpha^t_{\Lambda_n}(A)-\alpha^t(A)\Vert\leq\epsilon$. Then we have
\[
\Vert\alpha^t(A)-A\Vert \leq 3\epsilon+ \Vert\alpha^t_{\Lambda_n}(A_\epsilon)-A_\epsilon\Vert
\]
for any $t\in[-1,1]$. Hence for any sequence $\{t_m\}$ in $\mathbb{R}$ with $\lim_{m\to\infty}t_m=0$, we have 
\[
0\leq \limsup_{m\to\infty}\Vert\alpha_{t_m}(A)-A\Vert
\leq 3\epsilon+\lim_{m\to\infty} \Vert\alpha^{t_m}_{\Lambda_n}(A_\epsilon)-A_\epsilon\Vert =3\epsilon. 
\]
Hence $\alpha^t$ is strongly continuous.
\end{proof}

\begin{proof}[Proof of Proposition \ref{prop_derivation}]
Assume that $A\in\mathcal{A}_X$. By assumption (A1), we have $\delta(A)=\delta_\Lambda(A)$ for any $\Lambda\in\mathcal{P}_0(\Gamma)$ with $d(X,\Lambda^c)>\mathfrak{r}$. Hence (i) has been shown. 

To prove (ii), assume that $A\in\mathcal{A}_X$. By assumption (A1),
\[
\Vert\delta^n(A)\Vert 
\leq 2^n\Vert A\Vert \sum_{\substack{Z_n\in\mathcal{P}_0(\Gamma) \\Z_n\cap S_{n-1}\neq\emptyset}} \cdots \sum_{\substack{Z_1\in\mathcal{P}_0(\Gamma) \\Z_1\cap S_0\neq\emptyset}} \prod_{k=1}^n\Vert\Phi(Z_k)\Vert
\]
holds, where we write 
\[
S_0:=X,\ S_1:=Z_1\cup S_0,\ldots,\ S_n:=Z_n\cup S_{n-1}. 
\]
Then we observe that
\begin{align*} \sum_{\substack{Z_n\in\mathcal{P}_0(\Gamma) \\Z_n\cap S_{n-1}\neq\emptyset}} \cdots \sum_{\substack{Z_1\in\mathcal{P}_0(\Gamma) \\Z_1\cap S_0\neq\emptyset}} \prod_{k=1}^n\Vert\Phi(Z_k)\Vert
&\leq \sum_{z_1\in S_0}\cdots\sum_{z_n\in S_{n-1}} \sum_{\substack{Z_1\in\mathcal{P}_0(\Gamma) \\z_1\in Z_1}} \cdots \sum_{\substack{Z_n\in\mathcal{P}_0(\Gamma) \\z_n\in Z_n}} \prod_{k=1}^n\Vert\Phi(Z_k)\Vert\\
&\leq \mathfrak{j}'^n \sum_{z_1\in S_0}\cdots\sum_{z_n\in S_{n-1}} 1 \leq \mathfrak{j}'^n \prod_{k=1}^n(Nk +|X|)
\end{align*}
where we write $N=N_\Phi:= \sup_{x\in\Gamma}|\mathfrak{B}_\mathfrak{r}(x)|$. Thus we have 
\[
\left\Vert\frac{t^n}{n!}\delta^n(A) \right\Vert
\leq \frac{(2\mathfrak{j}'|t|)^n}{n!} \Vert A\Vert \prod_{k=1}^n (Nk+|X|)
\leq \big(2\mathfrak{j}'|t|(N+|X|)\big)^n \Vert A\Vert.
\]
One can find a $T>0$ such that $2\mathfrak{j}'T(N+|X|)<1$. Then 
\[
\exp(t\delta)(A):=\sum_{n=0}^\infty\frac{t^n}{n!}\delta^n(A)
\]
converges in norm for any $t\in(-T,T)$. Furthermore the convergence is uniformly on any compact subsets. 

Thanks to Proposition \ref{prop_time}, the first equality in (iii) holds. Moreover, both $\alpha^t_\Lambda(A)$ and $\exp(t\delta_\Lambda)(A)$ are solution to the initial value problem
\[
\frac{d}{dt}A(t)=i[H_\Lambda,A(t)],\quad A(0)=A. 
\]
Since the solution is unique, the second equality in (iii) is derived. To see that the third equality also holds, let $\{\Lambda_m\}$ be a monotone increase sequence of elements of $\mathcal{P}_0(\Gamma)$ with $\bigcup_{m=1}^\infty\Lambda_m=\Gamma$. For any $A\in\mathcal{A}_\mathrm{loc}$ and for any $t\in(-T,T)$, we have 
\[
\Vert\exp(t\delta)(A)- \exp(t\delta_{\Lambda_m})(A)\Vert 
\leq \sum_{n=N(\Lambda_m)}^\infty \frac{|t|^n}{n!} \Vert\delta^n(A)-\delta_{\Lambda_m}^n(A)\Vert,
\]
where we write
\[
N(\Lambda_m):=\min\{n\in\mathbb{N}\,;\, \delta^n(A)=\delta_{\Lambda_m}^n(A)\}.
\]
Then $N(\Lambda_m)\to\infty$ as $m\to\infty$. By Fubini's theorem, we have $\exp(t\delta_{\Lambda_m})(A)\to \exp(t\delta)(A)$ as $m\to\infty$. Furthermore the convergence is uniform on any compact subsets in $(-T,T)$. 
\end{proof}

\end{document}